\documentclass[a4paper,twoside,12pt]{article}
\usepackage{amsmath,bm}
\usepackage{amsfonts}
\usepackage{amssymb}
\usepackage{mathrsfs}
\usepackage{colortbl}
\usepackage{geometry}
\usepackage{times}
\usepackage{enumerate}
\usepackage[pdfborder=false]{hyperref}
\usepackage{fullpage}
\usepackage[T1]{fontenc}
\usepackage{ae,aecompl}
\usepackage[round,authoryear]{natbib}
\usepackage{tikz}
\usepackage{color}
\usepackage{array}
\usepackage{graphicx}
\usepackage{epsfig}
\usepackage{slashbox}
\usepackage{lscape}
\usepackage{epic}
\usepackage{fixltx2e}
\usepackage{extarrows}
\usepackage{arcs}
\usepackage{ntheorem}
\usepackage{graphics}
\usepackage{tabularx}
\usepackage{makecell}
\usepackage{mathtools}
\usepackage{stmaryrd}

\makeatletter
\def\overrightharpoonfill@{\arrowfill@\relbar\relbar\rightharpoonup}
\DeclareRobustCommand{\overrightharpoon}{\mathpalette{\overarrow@\overrightharpoonfill@}}

\def\downrightharpoonfill@{\arrowfill@\relbar\relbar\rightharpoondown}
\DeclareRobustCommand{\downrightharpoon}{\mathpalette{\raise0.2em\underarrow@\downrightharpoonfill@}}
\makeatother

\definecolor{webgreen}{rgb}{0,0.4,0}
\definecolor{webbrown}{rgb}{0.6,0,0}
\definecolor{purple}{rgb}{0.5,0,0.25}
\definecolor{darkblue}{rgb}{0,0,0.7}
\definecolor{darkred}{rgb}{0.7,0,0}
\hypersetup{colorlinks,citecolor=darkred,filecolor=black,linkcolor=darkblue,urlcolor=webgreen,pdfpagemode=None,pdfstartview=FitH}
\renewcommand{\cite}{\citet}
\makeatletter
\newcommand{\ignore}[1]{}
\newtheorem{lemma}{{\sc Lemma}}
\newtheorem{proposition}{{\sc Proposition}}
\newtheorem{corollary}{{\sc Corollary}}
\newtheorem{theorem}{{\sc Theorem}}
\newtheorem{definition}{{\sc Definition}}

\newtheorem{example}{{\sc Example}}
\newtheorem{remark}{{\sc Remark}}
\newenvironment{proof}{\noindent {\bf \sl Proof\/}:\enspace}
{\hfill $\Box$ \vspace{0.5em}}

\renewcommand\section{\@startsection {section}{1}{\z@}{-1ex \@plus -1ex \@minus -.2ex}                                   {1ex }{\centering\large\scshape}}

\renewcommand\subsection{\@startsection {subsection}{1}{\z@}{-1ex \@plus -1ex \@minus -.2ex}                                   {1ex }{\raggedright\large\scshape}}

\begin{document}

\title{\sc \Large Decomposability and Strategy-proofness\\
in Multidimensional Models\thanks{%
Huaxia Zeng acknowledges that his work was supported by the
Program for Professor of Special Appointment (Eastern Scholar) at Shanghai Institutions of Higher Learning
(No.~2019140015).}}
\author{Shurojit Chatterji\thanks{%
School of Economics, Singapore Management University}~ and
Huaxia Zeng\thanks{%
School of Economics, Shanghai University of Finance and Economics, and the Key Laboratory of Mathematical Economics (SUFE), Ministry of Education, China}}
\date{\today }
\maketitle

\begin{abstract}
\noindent
We introduce the notion of a multidimensional hybrid preference domain on a (finite) set of alternatives that is a Cartesian product of finitely many components.
We demonstrate that in a model of public goods provision, multidimensional hybrid preferences arise naturally through assembling marginal preferences under the condition of semi-separability -
a weakening of separability.
The main result shows that under a suitable ``richness'' condition,
every strategy-proof rule on this domain can be decomposed into component-wise strategy-proof rules, and
more importantly every domain of preferences that reconciles decomposability of rules with strategy-proofness must be a multidimensional hybrid domain.


\end{abstract}
\bigskip\medskip

\section{Introduction}\label{sec:Introduction}

Public decisions entail vast expenditures on a variety of components such as defence, education, health.
A mechanism design approach would base the decisions on the set of alternatives,  formulated as the Cartesian product of these multiple components (denoted $A\coloneqq\times_{s \in M} A^s$),  on the preferences of agents over $A$.
Decision making in such multidimensional settings is considerably more tractable if it can be ``decomposed'', that is, if the social planner is able to take the decisions on each of the components  independently based on ``marginal'' preferences in each component that are derived from agents' ``overall'' preferences, and then piece these component-wise decisions into a final social decision.
Of course, one would also like this decomposed decision making process to have nice incentive properties.
Thus we seek to study ``straightforward''  mechanisms in multidimensional settings, that is, mechanisms that are decomposable and strategy-proof.

If overall preferences satisfy \emph{separability}\footnote{An overall preference is separable if
a marginal preference on each component can be induced such that
for any two alternatives, the one endowed with a better element at each disagreed component is preferred.
Separability is an important preference restriction widely investigated in both the literature on strategic voting \citep[e.g.,][]{LS1999} and
on mechanism design with monetary compensations \citep[e.g.,][]{R1979}.},
a straightforward mechanism for social decisions can be simply constructed by assembling independent component-wise mechanisms that are also strategy-proof.
However, separability is too demanding; for instance, in a model of club member recruitment  \citep[see][]{BSZ1991}, one might imagine that while the appointment of exactly one candidate is preferred to nobody being appointed, it may be less desirable  to recruit all candidates, while in an auction model with non-quasilinear preferences \citep[see][]{MS2015},
large-scale payments might influence an agent's ability to utilize objects.
With a view to broadening the scope of straightforward mechanism design in multidimensional settings, we seek to develop here a methodology that accommodates non-separable preferences that go beyond multidimensional single-peakedness of \citet{BGS1993}, and allows us to answer the  following question:
What sort of preference domain over alternatives (formulated as a Cartesian product of multiple components)  reconciles decomposability with strategy-proofness, that is, predicated on some way of deriving marginal preferences on each component,
(i) every rule\footnote{We focus on strategy-proofness, wherein a direct mechanism reduces to a social choice function that picks an outcome from $A$ for each preference profile. The social choice function will be assumed to satisfy the mild requirement of unanimity and we use the term ``rule'' to refer to a unanimous social choice function.} that is strategy-proof turns out to be decomposable into strategy-proof component-wise rules over marginal preferences, and
(ii) conversely, arbitrary strategy-proof component-wise rules can be assembled into a strategy-proof rule?

Once non-separable preferences are involved, the derivation of marginal preferences on a component may vary with the specification of elements on the remaining components.
Specifically, as proposed by \citet{LW1999},
fixing an arbitrary alternative $z$, one can induce a marginal preference in a component $s$ from an overall preference by eliciting the relative rankings of alternatives that share the same components $z^{-s}$.\footnote{In particular, when the overall preference is a separable preference, a unique marginal preference in each component is derived no matter which alternative $z$ is referred to.}
This then would affect the scope for designing strategy-proof component-wise rules, which would in turn
affect the class of decomposable, strategy-proof rules.\footnote{If ``too many'' marginal preferences are derived, only dictatorships on each component survive strategy-proofness. This implies that any strategy-proof rule other than a \emph{generalized dictatorship} (intuitively speaking, a combination of dictatorships on all components) fails to be decomposable, and consequently the scope for assembling strategy-proof component-wise rules is limited to the class of generalized dictatorships.\label{footnote:generalizeddictatorship}}
We propose a natural and consistent way of deriving marginal preferences from an overall preference $P_i$:
given an overall preference, refer to the top-ranked alternative ($a \coloneqq (a^s, a^{-s})$) and derive the marginal preference over a pair of elements $x^s$ and $y^s$ by comparing $(x^s,a^{-s})$ with $(y^s,a^{-s})$ in $P_i$.
We show that this way of deriving marginal preferences has the merit of precipitating the decomposability property on all strategy-proof rules on a general class of multidimensional models where preference domains satisfy a condition called \emph{multidimensional hybridness} (see an informal introduction in Section \ref{sec:heuristicexample}).
This hence allows social decisions to be made component wise and then assembled, thereby simplifying the task confronting the social planner.
Our way of deriving marginal preferences is mainly inspired by the fact that
the top-ranked alternative does play the role of an important benchmark in the specification of preference restrictions (recall the seminal preference restriction of single-peakedness) and in the study of strategy-proof rules
which in many models \citep[e.g.,][]{CS2011} are completely and endogenously determined by the profiles of top-ranked alternatively, i.e.,
satisfy \emph{the tops-only property}.

We make the following claims on domains of multidimensional hybrid preferences. First,
the notion of multidimensional hybridness allows for more flexible descriptions than separable and multidimensional single-peaked preferences respectively (see an illustration in Section \ref{sec:heuristicexample}).
Hence, multidimensional hybrid domains variously contain separable preferences, multidimensional single-peaked preferences and top-separable preferences of \citet{LW1999}.
Next, we demonstrate in a heuristic example of public goods provision in Section \ref{sec:heuristicexample},
that requiring ``semi-separability'' - a weakening of separability, in the procedure
wherein the information of hybridness restriction \citep[introduced by][]{CRSSZ2022} is extracted from the domains of marginal preferences and embedded into the overall preferences, provides an intuitive route to the generation of multidimensional hybrid preferences.
Finally, we show that on a class of rich domains (see the details in Section \ref{sec:theorem}), multidimensional hybrid domains are the unique ones that reconcile decomposability with strategy-proofness (see Theorem \ref{thm}).
This in return enables us to fully characterize strategy-proof rules on a rich multidimensional hybrid domain (see Corollary \ref{cor:characterization}),
so that earlier characterization results on the separable domain, the multidimensional single-peaked domain and the top-separable domain emerge as special cases of our analysis.
A key step in our analysis is establishing endogenously the tops-only property for all strategy-proof rules on a rich domain, and for marginal rules on domains of marginal preferences\footnote{In many cases, the tops-only property is necessary for the decomposability of strategy-proof rules \citep[e.g.,][]{BSZ1991,BGS1993,LW1999}.},
which clearly further simplifies the task of the designer since the social decision on each component is determined by the profile of peaks on that component.

This paper is organized as follows.
Section \ref{sec:heuristicexample} provides a heuristic example of public goods provision to illustrate how multidimensional hybrid preferences arise.
Section \ref{sec:preliminaries} sets out the model and preliminaries.
In Section \ref{sec:result}, we formally introduce multidimensional hybrid domains and establish the characterization results, while Section \ref{sec:conclusion} contains some final remark and a review of the literature.
The proof of the main theorem is contained in the Appendix, while all other omitted proofs are put in the Supplementary Material.

\section{A heuristic example}\label{sec:heuristicexample}

Multidimensional hybrid preferences are generalizations of the hybrid preferences of \citet{CRSSZ2022}, and
can arise naturally through assembling hybrid preferences under the condition of semi-separability.
We provide the following example in the model of public goods provision to illustrate.

\begin{figure}[t]

\hspace{1em}
\includegraphics[width=0.9\textwidth]{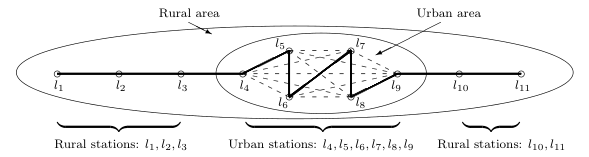}\\
\vspace{-3em}
\caption{\small The transportation system in the region\protect \footnotemark}\label{fig:region}
\end{figure}

\footnotetext{In Figure \ref{fig:region}, the bold line represents the railway, while the dashed lines denote the metro transportation system in the urban area that complements the railway.}

Imagine a railway running in a region, which contains several stations $\Omega = \{l_1, \dots, l_v\}$,
$v \geq 2$.
An urban area stands in the center of the region, surrounded by a large rural area.
The railway goes through the urban area, and all multiple urban stations cluster in the middle.
The urban area is postulated to possess a modern metro transportation system that fully connects all urban stations (see for instance Figure \ref{fig:region}).
Thus, we can identify two particular urban stations $l_{\underline{k}}$ and $l_{\overline{k}}$,
where $\underline{k}< \overline{k}$ (see $l_4$ and $l_9$ in Figure \ref{fig:region}),
that separate the railway into three parts:
left rural stations $\mathcal{L} = \{l_1, \dots, l_{\underline{k}-1}\}$,
middle urban stations $\mathcal{M} = \{l_{\underline{k}}, \dots, l_{\overline{k}}\}$ and
right rural stations $\mathcal{R} = \{l_{\overline{k}+1}, \dots, l_v\}$.
As $l_{\underline{k}}$ is directly connected to all urban stations and is the gate to the left rural stations, it can be viewed as the left transportation hub. Symmetrically, $l_{\overline{k}}$ is the right transportation hub.

A set of multiple good public facilities $M = \{1, \dots, m\}$, $m \geq 2$,
like a sports complex and a shopping mall, needs to be allocated on stations.
For each public facility $s$, let a nonempty and non-singleton subset $A^s \subseteq \Omega$ collect locations that are available for construction.
Clearly, all locations of $A^s$ are linearly ordered on the railway, i.e.,
given $a^s = l_p\in A^s$ and $b^s = l_q\in A^s$, $[a^s \prec^s b^s] \Leftrightarrow [p< q]$.
Given $x^s, y^s \in A^s$,
let $\textrm{Int}\langle x^s, y^s\rangle \coloneqq \big\{a^s \in A^s: x^s \prec^s a^s \prec^s y^s\; \textrm{or}\; y^s \prec^s a^s \prec^s x^s\big\}$ collect feasible locations that are located strictly between $x^s$ and $y^s$.
Furthermore, for simplicity, we assume that both transportation hubs are available for the construction of the public facility $s$ whenever at least two urban locations are included in $A^s$, i.e., $\big[|A^s\cap \mathcal{M}|\geq 2\big]\Rightarrow \big[l_{\underline{k}}, l_{\overline{k}} \in A^s\big]$.
Henceforth, for each public facility $s \in M$, we identify two particular feasible locations $\underline{x}^s$ and $\overline{x}^s$, called \emph{threshold locations}, such that
$\big[|A^s\cap \mathcal{M}| \leq 2\big] \Rightarrow \big[\underline{x}^s=\overline{x}^s \in A^s\; \textrm{is arbitrary}\big]$ and
$\big[|A^s\cap \mathcal{M}| > 2\big] \Rightarrow \big[\underline{x}^s= l_{\underline{k}}\; \textrm{and}\;\overline{x}^s =l_{\overline{k}}\big]$.
Clearly, an $m$-tuple $(a^1, \dots, a^m) \in \times_{s \in M}A^s$ represents a feasible allocation of these $m$ public facilities.

For each public facility $s$,
an individual $i$ living in this region has a marginal preference $P_i^s$ over $A^s$.
Each individual's marginal preference is private information, and
the social planner only knows that a domain $\mathbb{D}^s$ contains all individuals' marginal preferences.
It is natural to assume here that
the formulation of an individual's marginal preference is determined by the
distance between locations measured by both the railway and the metro transportation system.
Specifically,
if $|A^s\cap \mathcal{M}| \leq 2$,
an individual would prefer an available location that is closer to his/her own location along the railway, and therefore has \emph{single-peaked} preferences on $A^s$ w.r.t.~$\prec^s$; if $|A^s\cap \mathcal{M}| > 2$,
an individual living around a rural station would have single-peaked preferences on available locations that lie at the two sides of the transportation hubs along the railway, and
prefer the proximate transportation hub to all other available locations in the urban area,
while an individual living in the urban area would have arbitrary preferences on all urban available locations attributed to the complementary metro transportation, and prefer an rural available location that is closer to its proximate transportation hub.
Overall, we summarize that each marginal preference $\hat{P}_i^s \in \mathbb{D}^s$, say that $x^s$ is the top-ranked location, is \textbf{hybrid} on $\prec^s$ w.r.t.~$\underline{x}^s$ and $\overline{x}^s$,
i.e., given two distinct locations $a^s, b^s \in A^s$,
\begin{align*}
\big[a^s \in \textrm{Int}\langle x^s, b^s\rangle\;\textrm{and}\; a^s \notin \textrm{Int}\langle \underline{x}^s, \overline{x}^s\rangle\big]
\Rightarrow \big[a^s\mathrel{\hat{P}_i^s}b^s\big].
\end{align*}

Each individual $i$ also has a private overall preference $P_i$ over all feasible allocations
$A \coloneqq \times_{s \in M}A^s$,
which is private information as well.\footnote{Given $s \in M$, let $A^{-s} \coloneqq \mathop{\times}\limits_{t \in M\backslash \{s\}}A^t$ denote the set of all feasible allocations for public facilities other than $s$.}
The social planner of course will make some inference on individuals' overall preferences based on the known information. Specifically, the social planner believes that
each individual's overall preference is formulated according to the condition of \textit{semi-separability}
which requires the rankings of allocations to fully respect the preference restrictions embedded in all marginal domains $\mathbb{D}^1, \dots, \mathbb{D}^m$.
Formally,
an overall preference $P_i$, where the allocation $x$ is top-ranked,
is \textbf{semi-separable} if for all distinct allocations $a,b \in A$,
we have
\begin{align*}
\left[
\begin{array}{l}
a^s \mathrel{\hat{P}_i^s} b^s\; \textrm{for all}\; s \in M\; \textrm{such that}\; a^s\neq b^s\;
\textrm{and}\\
\textrm{all}\; \hat{P}_i^s \in \mathbb{D}^s\; \textrm{such that}\; x^s\; \textrm{is top-ranked}
\end{array}
\right]
\Rightarrow \big[a\mathrel{P_i}b\big].\footnotemark
\end{align*}

\footnotetext{Indeed, semi-separability weakens the restriction of separability.
For instance, given $s \in M$, pairwise distinct $x^s, a^s, b^s \in A^s$ and
$\hat{P}_i^s, \tilde{P}_i^s \in \mathbb{D}^s$ such that $r_1(\hat{P}_i^s) = r_1(\tilde{P}_i^s) = x^s$,
$a^s \mathrel{\hat{P}_i^s} b^s$ and $b^s \mathrel{\tilde{P}_i^s} a^s$,
a semi-separable preference $P_i$ that includes $x^s$ in its top-ranked allocation,
can have $(a^s, y^{-s})\mathrel{P_i} (b^s, y^{-s})$ and
$(b^s, z^{-s})\mathrel{P_i} (a^s, z^{-s})$ simultaneously.}

\noindent
Immediately,
we realize that the overall preference $P_i$, recalling that $x$ is top-ranked, satisfies
the following restriction:
for all allocations $a,b \in A$ that disagree on exactly one dimension, say $a^s \neq b^s$ and $a^{-s} = b^{-s}$,
$a$ is strictly preferred to $b$ whenever
either $a^s$ equals $x^s$,
or $a^s$ is located strictly between $x^s$ and $b^s$, but \emph{not} strictly between $\underline{x}^s$ and $\overline{x}^s$, i.e.,
\begin{align*}
	&~\big[\textrm{either}\; a^s = x^s,\;
	\textrm{or}\; a^s \in \textrm{Int}\langle x^s, b^s\rangle\;\textrm{and}\; a^s \notin \textrm{Int}\langle \underline{x}^s, \overline{x}^s\rangle
	\big] \\
    \xLongrightarrow{\textrm{~induce hybridness~~}}&
    ~\big[a^s \mathrel{\hat{P}_i^s} b^s\; \textrm{for all}\; \hat{P}_i^s \in \mathbb{D}^s\; \textrm{such that}\; x^s\; \textrm{is top-ranked}\big] \\
    \xLongrightarrow{\textrm{~via semi-separability~~}}
    & ~\big[a\mathrel{P_i}b\big].
\end{align*}
We call such a preference restriction \textbf{multidimensional hybridness},
as it incorporates the restriction of hybridness embedded in the marginal domain of each public facility
in formulating relative rankings of allocations of all $m$ public facilities.
Note that in the extreme case that $\underline{x}^s = \overline{x}^s$ for all $s \in M$,
multidimensional hybridness is strengthened to the conventional restriction of multidimensional single-peakedness.
Therefore, the choice of distinct threshold allocations provides freedom to rank alternatives in ways that go beyond the requirement of multidimensional single-peakedness; for instance,
in the aforementioned multidimensional hybrid preference $P_i$,
given $a^s \in \textrm{Int}\langle x^s, b^s\rangle$ and $a^s \in \textrm{Int}\langle \underline{x}^s, \overline{x}^s\rangle$,
we may simultaneously have $(a^s, y^{-s})\mathrel{P_i}(b^s, y^{-s})$ and $(b^s, z^{-s})\mathrel{P_i}(a^s, z^{-s})$, which of course also indicate a violation of separability.


In the remainder of the paper, we establish a general multidimensional model,
where multidimensional hybrid preferences are formulated without the imposition of any additional condition like semi-separability, and explore the salience of multidimensional hybrid domains by characterizing
that under some mild richness condition, multidimensional hybridness is necessary and sufficient for a  preference domain to reconcile decomposability with strategy-proofness.

\newpage
\section{Preliminaries}
\label{sec:preliminaries}

Let $A$ be a finite set of alternatives.
We throughout the paper assume
that the alternative set is represented by a \textbf{Cartesian product}
of a finite number of sets, each of which contains finitely many elements.
Formally, we fix $A = \times_{s \in M}A^{s}$ where $M= \{1, \dots, m\}$, $m \geq 2$ is an integer, and $2 \leq |A^{s}| < \infty$ for each $s \in M$.\footnote{The condition $|A^s| \geq 2$ ensures indispensability of the component $s$.}
Here, each $s$ is called a \textbf{component}; $A^{s}$ is referred to as a \textbf{%
component set}, and an element in $A^{s}$ is denoted as $a^{s}$.
An alternative is represented by an $m$-tuple,
i.e., $a \coloneqq (a^1, \dots, a^{m})= (a^{s}, a^{-s})$.
Given $s \in M$ and $x^{-s}\in A^{-s}$,
let $(A^s, x^{-s}) \coloneqq \{a \in A: a^{-s} = x^{-s}\}$.
Given two alternatives $a,b\in A$, let $M(a,b) \coloneqq \{s \in M: a^s \neq b^s\}$ denote the set of components on which $a$ and $b$ disagree.
In particular, $a$ and $b$ are said \textbf{similar}
if $|M(a,b)| = 1$.
Let $N \coloneqq \{1, \dots, n\}$ be a finite set of voters with $n \geq 2$.
Each voter $i$ has a preference order $P_{i}$ over $A$ which is complete,
antisymmetric and transitive, i.e., a \emph{linear order}.
For any $a, b \in A$,
$a\mathrel{P_{i}}b$ is interpreted as ``$a$ is strictly preferred to $b$
according to $P_{i}$".
Given a preference $P_{i}$, let $r_{k}(P_{i})$, where $1 \leq k \leq |A|$, denote the $k$th ranked
alternative in $P_{i}$.
Moreover, given a nonempty subset $B \subseteq A$, let $\max^{P_i}(B)$ and $\min^{P_i}(B)$ be the best and worst alternatives in $B$ according to $P_i$ respectively.
Two preferences $P_i$ and $P_i'$ are called \textbf{complete reversals} if for all $a,b \in A$, we have
$[a\mathrel{P_i}b] \Leftrightarrow [b\mathrel{P_i'}a]$.
Let $\mathbb{P}$
denote the set of \emph{all} linear orders over $A$.
The set of admissible preferences is a set $\mathbb{D} \subseteq \mathbb{P}$, referred
to as a \textbf{preference domain}.\footnote{In this paper, $\subseteq$ and $\subset$ denote the weak and strict inclusion relations respectively.}
We call $\mathbb{P}$ \textbf{the universal domain}.
Henceforth, each domain $\mathbb{D}$ under investigation is assumed to be \textbf{minimally rich},
i.e., for each $a\in A$, there exists $P_i \in \mathbb{D}$ such that $r_1(P_i) = a$.
A \textbf{preference profile} is an $n$-tuple $P \coloneqq (P_{1}, \dots, P_{n})
= (P_{i}, P_{-i}) \in \mathbb{D}^{n}$.
Analogously, for each $s \in M$,
let $P_i^s$ denote a \textbf{marginal preference} over $A^s$, $\mathbb{P}^s$ denote  \textbf{the universal marginal domain}, and
$\mathbb{D}^s \subseteq \mathbb{P}^s$ denote an admissible \textbf{marginal domain}.

A \textbf{Social Choice Function} (or \textbf{SCF}) is a map $f: \mathbb{D}^{n} \rightarrow A$,
which associates to each preference profile $P \in \mathbb{D}^{n}$, a
``socially desirable'' outcome $f(P)$.
First, an SCF $f: \mathbb{D}^{n}
\rightarrow A$ is required to be \textbf{unanimous}, i.e.,
for all $a \in A$ and $P \in \mathbb{D}^n$,
we have
$[r_{1}(P_{i}) = a$ for all $i \in N] \Rightarrow [f(P) = a]$.
For ease of presentation, a unanimous SCF henceforth is called a \textbf{rule}.
Next, an SCF $f: \mathbb{D}^n \rightarrow A$ satisfies the \textbf{tops-only property} if
for all $P,P' \in \mathbb{D}^n$, we have $[r_1(P_i) = r_1(P_i')\; \textrm{for all}\; i \in N] \Rightarrow
[f(P) = f(P')]$.
Last, an SCF $f: \mathbb{D}^n \rightarrow A$ is \textbf{%
strategy-proof} if for all $i \in N$, $P_{i}, P_{i}^{\prime } \in \mathbb{D}$
and $P_{-i} \in \mathbb{D}^{n-1}$, we have $[f(P_i,P_{-i}) \neq f(P_i', P_{-i})] \Rightarrow [f(P_i,P_{-i}) \mathrel{P_i} f(P_i', P_{-i})]$.
Analogously, given $s \in M$ and $[\mathbb{D}^s]^n \coloneqq \underset{n}{\underbrace{\mathbb{D}^s \times \dots \times \mathbb{D}^s}}$, a \textbf{marginal SCF} is a map $f^s: [\mathbb{D}^s]^n \rightarrow A^s$.
These three axioms alluded to also apply to marginal SCFs.
A unanimous marginal SCF is henceforth called a \textbf{marginal rule}.

\medskip
\subsection{\rm Separable preference and non-separable preference}\label{sec:separable}


Formally,
a preference $P_i$ is \textbf{separable} if there exists a marginal preference $P_i^s$ for each $s \in M$ such that
for each pair of similar alternatives $a, b \in A$,
say $M(a,b) = \{s\}$, we have
$\big[a^{s}\mathrel{P_{i}^{s}} b^{s}\big] \Rightarrow \big[a\mathrel{P_i}b\big]$.
Let $\mathbb{D}_{\textrm{S}}$ denote \textbf{the separable domain} that contains all separable preferences.
Clearly, $\mathbb{D}_{\textrm{S}} \subset \mathbb{P}$, and
a preference that is not separable is called a non-separable preference.
Henceforth, a domain is said to satisfy \textbf{diversity\textsuperscript{$+$}} if it contains two
separable preferences that are complete reversals.\footnote{The term ``diversity\textsuperscript{$+$}''
strengthens the notion of \emph{diversity} introduced by \citet{CRSSZ2022}, as it further requires the complete reversals to be separable preferences.}

More importantly, we introduce a particular way of deriving marginal preferences from both separable and non-separable preferences.
Given a preference $P_i$ (separable or non-separable), say $r_1(P_i) = a$,
for each $s \in M$, referring to $a^{-s}$ which are contained in the peak of $P_i$,
we induce a marginal preference, denoted $[P_i]^s$, such that
for all $x^s, y^s \in A^s$, $\big[x^s\mathrel{[P_i]^s}y^s\big] \Leftrightarrow \big[(x^s, a^{-s})\mathrel{P_i}(y^s, a^{-s})\big]$.
Accordingly, let $[\mathbb{D}]^s \coloneqq \big\{[P_i]^s: P_i \in \mathbb{D}\big\}$ denote the set of marginal preferences induced from all preferences of $\mathbb{D}$.
To avoid confusion with the notation $P_i^s$ and  $\mathbb{D}^s$,
 we henceforth call $[P_i]^s$ an \textbf{induced marginal preference} and
$[\mathbb{D}]^s$ an \textbf{induced marginal domain}.\footnote{Indeed, both $P_i^s$ and $[P_i]^s$ refer to linear orders over $A^s$.
For the sake of notation, $[P_i]^s$ emphasizes that it is induced from a given preference $P_i$.
Similarly,
$[\mathbb{D}]^s$ emphasizes that it contains marginal preferences over $A^s$ that are induced from preferences of $\mathbb{D}$.}

\subsection{\rm Decomposable SCF and decomposable domain}

An SCF $f: \mathbb{D}^n \rightarrow A$ is said \textbf{decomposable} if for each $s \in M$,
there exists a marginal SCF $f^s: \big[[\mathbb{D}]^s\big]^n \rightarrow A^s$ such that
for all $(P_1, \dots, P_n) \in \mathbb{D}^n$, we have
\begin{align*}
\big[f(P_1, \dots, P_n) = a\big] \Leftrightarrow \big[f^s([P_1]^s, \dots, [P_n]^s) = a^s\; \textrm{for all}\;\, s \in M\big].
\end{align*}
We focus on preference domains that reconcile decomposability of all rules with strategy-proofness.


\begin{definition}\label{def:decomposabledomain}
A domain $\mathbb{D}$ is a \textbf{decomposable domain} if for every SCF $f: \mathbb{D}^n \rightarrow A$,
$n \geq 2$, we have
\begin{align*}
\Big[f \; \textrm{is a strategy-proof rule}\Big]
\Leftrightarrow
\left[
\begin{array}{l}
\!\!f\; \textrm{is decomposable, and all marginal SCFs}\\
\!\!f^1, \dots, f^m\; \textrm{are strategy-proof marginal rules}
\end{array}
\!\!\right].
\end{align*}
\end{definition}

On the one hand,
since dictatorships are strategy-proof marginal rules on arbitrary induced marginal domains,
by the requirement of the direction ``$\Leftarrow$'' in Definition \ref{def:decomposabledomain},
all generalized dictatorships (recall footnote \ref{footnote:generalizeddictatorship}) are entitled with strategy-proofness.
This implies that a decomposable domain must be embedded with some preference restriction.
For instance, if all preferences of the domain are separable,
an SCF constructed by assembling strategy-proof marginal rules immediately turns out to be a strategy-proof rule.
On the other hand, to meet the requirement of the direction ``$\Rightarrow$'' in Definition \ref{def:decomposabledomain}, a decomposable domain is required to contain sufficiently many preferences.
For instance, on the universal domain $\mathbb{P}$, by the Gibbard-Satterthwaite Theorem \citep{G1973,S1975},
each strategy-proof rule is a dictatorship, and hence can be decomposed into $m$ marginal dictatorships that share the same dictator.
Therefore, a decomposable domain must be a restricted preference domain that satisfies some richness condition.

\section{Results}\label{sec:result}
In this section, we introduce multidimensional hybrid domains, and adopt it to establish a complete characterization of decomposable domains under some mild richness condition.

\subsection{\rm Multidimensional hybrid domains and fixed ballot rules}\label{sec:MHDomain}

Fixing a linear order $\prec^s$ over $A^s$ for each $s \in M$,
let $\prec \,\coloneqq \times_{s \in M}\prec^s$ denote the Cartesian product of $\prec^1, \dots, \prec^m$.
Given $s \in M$ and $a^s, b^s \in A^s$, let $\langle a^s, b^s\rangle \coloneqq \{x^s \in A^s: a^s \preccurlyeq^s x^s \preccurlyeq^s b^s\; \textrm{or}\; b^s \preccurlyeq^s x^s \preccurlyeq^s a^s\}$ denote the set of elements that located between $a^s$ and $b^s$ on the linear order $\prec^s$,\footnote{For notational convenience, henceforth, let $a^s \preccurlyeq^s b^s$ denote either $a^s \prec^s b$ or $a^s = b^s$.} and
let $\textrm{Int}\langle a^s, b^s\rangle \coloneqq \{x^s \in A^s: a^s \prec^s x^s \prec^s b^s\; \textrm{or}\; b^s \prec^s x^s \prec^s a^s\}$ denote the set of elements that are located strictly between $a^s$ and $b^s$.
Given $s \in M$, two elements $\underline{x}^s$ and $\overline{x}^s$ are called \textbf{marginal thresholds} if either $\underline{x}^s = \overline{x}^s$,
or $\underline{x}^s \neq \overline{x}^s$ and $|\langle \underline{x}^s, \overline{x}^s\rangle|\geq 3$.
Correspondingly, two alternatives $\underline{x}$ and $\overline{x}$
are called \textbf{thresholds} if
for each $s \in M$, $\underline{x}^s$ and $\overline{x}^s$ are marginal thresholds.

\begin{definition}\label{def:MH}
A preference $P_i$, say $r_1(P_i) = x$, is \textbf{multidimensional hybrid} on $\prec$ w.r.t.~$\underline{x}$ and $\overline{x}$ if for all similar $a, b \in A$, say $M(a,b) = \{s\}$, we have
\begin{itemize}
\item[\rm (i)] $\big[a^s = x^s\big]
\Rightarrow \big[a\mathrel{P_i}b\big]$, and

\item[\rm (ii)] $\big[a^s \in \emph{Int}\langle x^s, b^s\rangle\; \textrm{and}\;\, a^s \notin \emph{Int}\langle \underline{x}^s, \overline{x}^s\rangle\big]
\Rightarrow \big[a\mathrel{P_i}b\big]$.\footnote{By transitivity, given distinct $a,b \in A$ (not necessarily similar alternatives),
if $a^s = x^s$, or $a^s \in \textrm{Int}\langle x^s, b^s\rangle\; \textrm{and}\; a^s \notin \textrm{Int}\langle \underline{x}^s, \overline{x}^s\rangle$ hold for all $s \in M(a, b)$, we have $a\mathrel{P_i}b$.\label{footnote}}
\end{itemize}
\end{definition}

\begin{remark}\label{rem:relation1}\rm
Consider two extreme cases:
(1) $\underline{x}^s = \min^{\prec^s}(A^s)$ and $\overline{x}^s = \max^{\prec^s}(A^s)$ for each $s \in M$, and
(2) $\underline{x} = \overline{x}$.
In the first case, condition (ii) of Definition \ref{def:MH} is redundant as $a^s \in \textrm{Int}\langle x^s, b^s\rangle$ and $a^s \notin \textrm{Int}\langle \underline{x}^s, \overline{x}^s\rangle$ cannot hold simultaneously.
Thus, only condition (i) survives, and hence
a multidimensional hybrid preference turns to be a \textbf{top-separable} preference of \citet{LW1999}.
In the second case,
since the hypothesis $a^s \notin \textrm{Int}\langle \underline{x}^s, \overline{x}^s\rangle = \emptyset$ is vacuously satisfied,
the two conditions of Definition \ref{def:MH} can be merged:
$\big[a^s \in \langle x^s, b^s\rangle\big] \Rightarrow [a\mathrel{P_i}b]$.
Then, a multidimensional hybrid preference becomes as restrictive as a \textbf{multidimensional single-peaked} preference of \citet{BGS1993}.
\end{remark}

We provide one example to explain multidimensional hybrid preferences.

\begin{figure}[t]
\begin{center}
\includegraphics[width=0.5\textwidth]{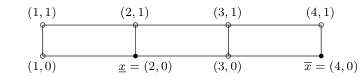}\\
\caption{The Cartesian product of linear orders $\prec= \prec^1 \times \prec^2$}\label{fig:grid}
\end{center}
\end{figure}

\begin{example}\label{exm:MH}\rm
Recall the model of public goods provision in Section \ref{sec:heuristicexample}.
Let $M = \{1,2\}$, $A^1 = \{l_2, l_4, l_6, l_9\}$ and $A^2 = \{l_3, l_8\}$.
Thus, we have the linear orders $l_2 \prec^1 l_4 \prec^1 l_6\prec^1 l_9$ and $l_3 \prec^2 l_8$,
two thresholds $\underline{x} \coloneqq (l_4, l_3)$ and $\overline{x} \coloneqq (l_9, l_3)$,
and
all feasible allocations $A\coloneqq A^1 \times A^2$ arranged on the grid $\prec = \prec^1 \times \prec^2$ in Figure \ref{fig:grid}.
For instance, we specify the restriction of multidimensional hybridness on a preference with the peak $(l_2,l_3)$.
First, to meet condition (i) of Definition \ref{def:MH},
it must be the case that
for each $a^1 \in \{l_4, l_6, l_9\}$,
$(l_2,l_8)$ is ranked above $(a^1,l_8)$, and $(a^1, l_3)$ is ranked above $(a^1, l_8)$.
Second, since $l_4 \in \textrm{Int}\langle l_2, l_6\rangle$, $l_4 \in \textrm{Int}\langle l_2, l_9\rangle$ and $l_4 \notin
\textrm{Int}\langle l_4, l_9\rangle =\textrm{Int}\langle \underline{x}^1, \overline{x}^1\rangle$,
to meet condition (ii) of Definition \ref{def:MH},
we need to ensure that
for each $a^2 \in A^2$,
$(l_4,a^2)$ is ranked above both $(l_6,a^2)$ and $(l_9,a^2)$.
It is worth mentioning that the relative ranking between $(l_6,l_3)$ and $(l_9,l_3)$ (also between $(l_6,l_8)$ and $(l_9,l_8)$) is allowed to be arbitrary.
Accordingly, we specify two examples of such multidimensional hybrid preferences:
\begin{align*}
P_i = &~ (l_2,l_3)_{\rightharpoonup}(l_4,l_3)_{\rightharpoonup}(l_6,l_3)_{\rightharpoonup}(l_9,l_3)_{\rightharpoonup}(l_2,l_8)_{\rightharpoonup}(l_4,l_8)_{\rightharpoonup}(l_6,l_8)_{\rightharpoonup}(l_9,l_8),\; \textrm{and}\\
\hat{P}_i = &~
(l_2,l_3)_{\rightharpoonup}(l_2,l_8)_{\rightharpoonup}(l_4,l_3)_{\rightharpoonup}(l_4,l_8)_{\rightharpoonup}(l_6,l_3)_{\rightharpoonup}(l_9,l_3)_{\rightharpoonup}(l_9,l_8)_{\rightharpoonup}(l_6,l_8).\protect \footnotemark
\end{align*}
Note that $P_i$ is separable, and $\hat{P}_i$ is a non-separable preference.
\footnotetext{To save space, we specify the two preferences here horizontally. For instance,
the notation ``$(l_2,l_3)_{\rightharpoonup}(l_4,l_3)$'' in the specification of $P_i$ represents ``$(l_2,l_3)\mathrel{P_i}(l_4,l_3)$''.}
\hfill$\Box$
\end{example}

%
%

We focus on a large family of domains of multidimensional hybrid preferences
where sufficiently many marginal preferences in each component can be induced.
To do so, we first introduce some standard concepts from graph theory.
An (undirected) \textbf{graph}, denoted $G \coloneqq \langle V, \mathcal{E}\rangle$, is a combination of a ``vertex set'' $V$ and an ``edge set'' $\mathcal{E} \subseteq V\times V$ such that
$\big[(\alpha, \beta) \in \mathcal{E}\big] \Rightarrow \big[\alpha \neq \beta\; \textrm{and}\;  (\beta, \alpha) \in \mathcal{E}\big]$.
A vertex $\alpha \in V$ is called a \textbf{leaf} if there exists a unique $\beta \in V$ such that $(\alpha, \beta) \in \mathcal{E}$.
Given $\alpha, \beta \in V$, a \textbf{path} in $G =\langle V, \mathcal{E}\rangle $ connecting $\alpha$ and $\beta$ is a sequence of non-repeated vertices $(\alpha_1, \dots, \alpha_v)$, $v \geq 2$, such that
$\alpha_1 = \alpha$, $\alpha_v = \beta$ and
$(\alpha_k, \alpha_{k+1}) \in \mathcal{E}$ for all $k = 1, \dots, v-1$.
A graph $G = \langle V, \mathcal{E}\rangle$ is a \textbf{connected graph} if
each pair of distinct vertices is connected by a path.
Note that the vertex set $V$ here can be
a subset of elements,
of alternatives,
of marginal preferences, or
of preferences.
For instance, following Definition 1 of \citet{CSS2013},
two elements $a^s, b^s \in A^s$ are said \textbf{strongly connected},
denoted $a^s \approx b^s$,
if there exist $[P_i]^s, [P_i']^{s} \in [\mathbb{D}]^s$ such that
$r_1([P_i]^s) = r_2([P_i']^s) = a^s$, $r_1([P_i']^s) = r_2([P_i]^s) = b^s$ and
$r_k([P_i]^s) = r_k([P_i']^s)$ for all $k = 3, \dots, |A^s|$.
Accordingly, given a nonempty subset $B^s \subseteq A^s$,
we can induce a graph $G_{\approx}^{B^s} \coloneqq \langle B^s, \mathcal{E}_{\approx}^{B^s}\rangle$
where two elements of $B^s$ form an edge if and only if they are strongly connected.

\begin{definition}\label{def:AMH}
A domain $\mathbb{D}$ is called \textbf{a multidimensional hybrid domain} if
there exist thresholds $\underline{x}, \overline{x} \in A$ such that
\begin{itemize}
\item[\rm (i)] all preferences of $\mathbb{D}$ are multidimensional hybrid on $\prec$ w.r.t.~$\underline{x}$ and $\overline{x}$, and

\item[\rm (ii)] for each $s \in M$,  $G_{\approx}^{A^s}$ is
a connected graph, and\\
$\big[s \in M(\underline{x}, \overline{x})\big] \Rightarrow
\big[G_{\approx}^{\langle \underline{x}^s, \; \overline{x}^s\rangle}\; \textrm{has no leaf}\,\big]$.
\end{itemize}
\end{definition}

Given thresholds $\underline{x}, \overline{x} \in A$,
it is evident that \textbf{the multidimensional hybrid domain} which contains all multidimensional hybrid preferences on $\prec$ w.r.t.~$\underline{x}$ and $\overline{x}$, denoted $\mathbb{D}_{\textrm{MH}}(\prec, \underline{x}, \overline{x})$, is a multidimensional hybrid domain.
It is easy to verify that the separable domain $\mathbb{D}_{\textrm{S}}$ is a multidimensional hybrid domain:
fixing thresholds $\underline{x}$ and $\overline{x}$ such that for each $s \in M$,
$[|A^s| = 2] \Rightarrow [\underline{x}^s = \overline{x}^s \in A^s]$ and
$[|A^s| \geq 3] \Rightarrow [\underline{x}^s = \min^{\prec^s}(A^s)\; \textrm{and}\; \overline{x}^s = \max^{\prec^s}(A^s)]$,
we have that (i) all preferences of $\mathbb{D}_{\textrm{S}}$ are multidimensional hybrid on $\prec$ w.r.t.~$\underline{x}$ and $\overline{x}$, and
(ii) for each $s \in M$, since $[\mathbb{D}_{\textrm{S}}]^s = \mathbb{P}^s$,
 $G_{\approx}^{A^s}$ is a connected graph and $\big[s \in M(\underline{x}, \overline{x})\big] \Rightarrow
\big[G_{\approx}^{\langle \underline{x}^s, \; \overline{x}^s\rangle} = G_{\approx}^{A^s}\; \textrm{has no leaf}\,\big]$.

Recall the two extreme cases in Remark \ref{rem:relation1}:
(1) $\underline{x}^s = \min^{\prec^s}(A^s)$ and $\overline{x}^s = \max^{\prec^s}(A^s)$ for each $s \in M$, and
(2) $\underline{x} = \overline{x}$.
In the first case, the multidimensional hybrid domain $\mathbb{D}_{\textrm{MH}}(\prec, \underline{x}, \overline{x})$ expands to the top-separable domain of \citet{LW1999}, denoted $\mathbb{D}_{\textrm{TS}}$,
while in the second case, $\mathbb{D}_{\textrm{MH}}(\prec, \underline{x}, \overline{x})$ shrinks to the multidimensional single-peaked domain of \citet{BGS1993}, denoted $\mathbb{D}_{\textrm{MSP}}(\prec)$.
Therefore, Definition \ref{def:AMH} is also applicable to these two important preference domains.
Last, note that for any thresholds $\underline{x}, \overline{x} \in A$ other than the two extreme cases,
it is true that $\mathbb{D}_{\textrm{TS}} \supset \mathbb{D}_{\textrm{MH}}(\prec,\underline{x}, \overline{x}) \supset \mathbb{D}_{\textrm{MSP}}(\prec)$.

Fixing a multidimensional hybrid domain $\mathbb{D}$ on $\prec$ w.r.t.~thresholds $\underline{x}$ and $\overline{x}$,
given $s\in M$,
by condition (i) of Definition \ref{def:AMH},
it is clear that each induced marginal preference is hybrid on $\prec^s$ w.r.t.~$\underline{x}^s$ and $\overline{x}^s$ (recall the definition in Section \ref{sec:heuristicexample}).
Consequently, in conjunction with condition (ii) of Definition \ref{def:AMH},
we observe that
if $\underline{x}^s = \overline{x}^s$, $G_{\approx}^{A^s}$ is a line over $A^s$,
while if $\underline{x}^s \neq \overline{x}^s$,
$G_{\approx}^{A^s}$ is simply a combination of a line between $\min^{\prec^s}(A^s)$ and $\underline{x}^s$, a connected subgraph $G_{\approx}^{\langle \underline{x}^s, \; \overline{x}^s\rangle}$ that has no leaf, and a line between $\overline{x}^s$ and $\max^{\prec^s}(A^s)$.
More importantly, by condition (ii) of Definition \ref{def:AMH},
sufficiently many hybrid marginal preferences are induced so that
we can establish a complete characterization of strategy-proof marginal rules,
using the class of fixed ballot rules introduced by \citet{M1980}.

A marginal SCF $f^s: [\mathbb{D}^s]^n \rightarrow A^s$ is a \textbf{Fixed Ballot Rule} (or \textbf{FBR}) on $\prec^s$ if there exists $b_J^s \in A^s$, called a \emph{fixed ballot}, for each coalition $J \subseteq N$, satisfying
\emph{ballot unanimity}, i.e., $b_{\emptyset}^s = \min^{\prec^s}(A^s)$ and $b_{N}^s = \max^{\prec^s}(A^s)$, and
\emph{monotonicity}, i.e., $[J \subset J' \subseteq N] \Rightarrow [b_J^s \preccurlyeq^s b_{J'}^s]$,
such that for all $(P_1^s, \dots, P_n^s) \in [\mathbb{D}^s]^n$,
\begin{align*}
f^s(P_1^s, \dots, P_n^s)  = \mathop{\mathop{\max}\nolimits^{\prec^s}}\limits_{J \subseteq N~~~\,}
\Big(\mathop{\mathop{\min}\nolimits^{\prec^s}}\limits_{i \in J~~~}\big(r_1(P_i^s), b_J^s\big)\Big).
\end{align*}
Furthermore, given $\underline{x}^s, \overline{x}^s \in A^s$ such that $\underline{x}^s \prec^s \overline{x}^s$,
the FBR $f^s$ is called an \textbf{$\bm{(\underline{x}^s, \overline{x}^s)}$-FBR}
\citep[introduced by][]{CRSSZ2022} if it in addition satisfies \emph{the constrained-dictatorship condition}, i.e., there exists $i \in N$ such that $[i \in J] \Rightarrow \big[\overline{x}^s \preccurlyeq^s b_J^s \big]$ and $[i \notin J] \Rightarrow \big[b_J^s \preccurlyeq^s \underline{x}^s\big]$.\footnote{The constrained-dictatorship condition ensures that the FBR $f^s$ \textbf{behaves like a dictatorship} on $\langle \underline{x}^s, \overline{x}^s\rangle$, i.e., for all $(P_1^s, \dots, P_n^s) \in [\mathbb{D}^s]^n$, $[r_1(P_1^s), \dots, r_1(P_n^s) \in \langle \underline{x}^s, \overline{x}^s\rangle] \Rightarrow \big[f^s(P_1^s, \dots, P_n^s) = r_1(P_i^s)\big]$.\label{footnote:behavelikeadictatorship}}

\begin{proposition}\label{prop:FBR}
Fix a multidimensional hybrid domain $\mathbb{D}$ on $\prec$ w.r.t.~thresholds $\underline{x}$ and $\overline{x}$.
Given $s\in M$, the following two statements hold:
\begin{itemize}
\item[\rm (i)] Given $\underline{x}^s = \overline{x}^s$,
an marginal SCF $f^s : \big[[\mathbb{D}]^s\big]^n \rightarrow A^s$ is a strategy-proof marginal rule if and only if it is an FBR.

\item[\rm (ii)] Given $\underline{x}^s \neq \overline{x}^s$,
an marginal SCF $f^s : \big[[\mathbb{D}]^s\big]^n \rightarrow A^s$ is a strategy-proof marginal rule if and only if it is an $(\underline{x}^s, \overline{x}^s)$-FBR.
\end{itemize}
\end{proposition}

The proof of the Proposition is lengthy, and hence is relegated to the Supplementary Material.

\subsection{\rm The Theorem}\label{sec:theorem}

In this section, we provide a complete characterization of decomposable domains under a mild richness condition.
We first introduce some necessary notions and notation for establishing the richness condition.
Two preferences $P_i$ and $P_i'$ are \textbf{adjacent}, denoted $P_i \sim P_i'$,
if there exist distinct $a, b \in A$ such that $r_k(P_i) = r_{k+1}(P_i') = a$ and $r_k(P_i') = r_{k+1}(P_i) = b$ for some $1\leq k < |A|$, and $r_{\ell}(P_i) = r_{\ell}(P_i')$ for all $\ell \notin \{k,k+1\}$.
As a natural extension of adjacency, two preferences $P_i$ and $P_i'$ are said \textbf{adjacent\textsuperscript{$+$}}, denoted $P_i \sim^{+} P_i'$, if they are separable preferences, and there exist $s \in M$ and distinct $a^s, b^s \in A^s$ such that
the following two conditions are satisfied:
\begin{itemize}
\item[\rm (i)] for all $z^{-s}\in A^{-s}$,
$(a^s, z^{-s}) = r_k(P_i) = r_{k+1}(P_i')$ and
$(b^s, z^{-s})=r_k(P_i') = r_{k+1}(P_i)$ for some $1\leq k < |A|$, and

\item[\rm (ii)] \makebox{for all $c \in A$, $\big[c^s \notin \{a^s, b^s\}\big] \Rightarrow
\big[c = r_{\ell}(P_i) = r_{\ell}(P_i')\;\textrm{for some} \; 1 \leq \ell \leq |A|\big]$.}
\end{itemize}
Given a domain $\mathbb{D}$,
we construct a graph $G_{\sim/\sim^+}^{\mathbb{D}}\coloneqq \langle \mathbb{D}, \mathcal{E}_{\sim/\sim^+}\rangle$
such that two preferences form an edge if and only if they are adjacent or adjacent\textsuperscript{+}.
Fixing a path $\pi = (P_{i|1}, \dots, P_{i|v})$ in $G_{\sim/\sim^+}^{\mathbb{D}}$,
given $a, b \in A$,
the path $\pi$ has \textbf{$\bm{\{a,b\}}$-restoration} if
the relative ranking of $a$ and $b$ has been flipped for more than once, i.e.,
there exist $1 \leq o< p < q \leq v$ such that
either $a\mathrel{P_{i|o}}b$, $b\mathrel{P_{i|p}}a$ and $a\mathrel{P_{i|q}}b$,
or $b\mathrel{P_{i|o}}a$, $a\mathrel{P_{i|p}}b$ and $b\mathrel{P_{i|q}}a$ hold.
\citet{S2013} restricted attention to the notion of adjacency,
and introduced \textbf{the no-restoration condition}:
given $P_i, P_i' \in \mathbb{D}$ and $a,b \in A$,
there exists a path in the graph $G_{\sim}^{\mathbb{D}}\coloneqq \langle \mathbb{D}, \mathcal{E}_{\sim}\rangle$, where two preferences form an edge if and only if they are adjacent, connecting $P_i$ and $P_i'$ that has no $\{a,b\}$-restoration.\footnote{Proposition 3.2 of \citet{S2013} has shown that the no-restoration condition is necessary for the equivalence between strategy-proofness and the notion of \emph{AM-proofness} which only prevents a voter's manipulation via misreporting preferences adjacent to the sincere one.}
Intuitively speaking, the no-restoration condition can be viewed as an ordinal counterpart of the convex-set assumption imposed on the valuation space in a cardinal model,
which reconciles the difference of any two preferences via a sufficiently short path.
The following two properties adopted from \citet{CZ2019} expand the no-restoration condition to the graph $G_{\sim/\sim^+}^{\mathbb{D}}$ that involves not only the edge of adjacency, but the edge of adjacency\textsuperscript{+} which is customized for separable preferences, and impose some additional requirements on some subgraphs of $G_{\sim/\sim^+}^{\mathbb{D}}$.

The Interior\textsuperscript{+} property concerns with preferences sharing the same peak, and
requires that such two preferences are connected by a path such that all preferences on the path have the same peak.

\begin{definition}
A domain $\mathbb{D}$ satisfies \textbf{the Interior\textsuperscript{$\bm{+}$} property}
if for all distinct $P_i, P_i' \in \mathbb{D}$ such that
$r_1(P_i) = r_1(P_i')\coloneqq x$,
there exists a path $\pi = (P_{i|1}, \dots, P_{i|v})$ in $G_{\sim/\sim^+}^{\mathbb{D}}$ connecting $P_i$ and $P_i'$ such that $r_1(P_{i|k}) = x$ for all $k = 1, \dots, v$.\footnote{This immediately implies that $P_i$, and $P_i'$ are connected by a path in $G_{\sim/\sim^+}^{\mathbb{D}}$ that has no $\{x,a\}$-restoration for any $a \in A\backslash \{x\}$.}
\end{definition}


The Exterior\textsuperscript{+} property concentrates on preferences with distinct peaks.
It imposes not only the no-restoration condition on $G_{\sim/\sim^+}^{\mathbb{D}}$,
but an additional condition on the path connecting any two preferences that have similar peaks.

\begin{definition}
A domain $\mathbb{D}$ satisfies \textbf{the Exterior\textsuperscript{$\bm{+}$} property}
if for all $P_i, P_i' \in \mathbb{D}$ such that $r_1(P_i) \neq r_1(P_i')$,
the following two conditions are satisfied:
\begin{itemize}
\item[\rm (i)] given $a,b \in A$,
there exists a path $\pi=(P_{i|1}, \dots, P_{i|v})$ in $G_{\sim/\sim^+}^{\mathbb{D}}$ connecting $P_i$ and $P_i'$ such that $\pi$ has no $\{a,b\}$-restoration, and

\item[\rm (ii)] (\textbf{the no-detour condition})
when $r_1(P_i)$ and $r_1(P_i')$ are similar, say $r_1(P_i), r_1(P_i') \in (A^s, x^{-s})$ for some $s \in M$ and $x^{-s} \in A^{-s}$,
there exists a path $\pi=(P_{i|1}, \dots, P_{i|w})$ in $G_{\sim/\sim^+}^{\mathbb{D}}$ connecting $P_i$ and $P_i'$ such that $r_1(P_{i|k}) \in (A^s, x^{-s})$ for all $k = 1, \dots, w$.
\end{itemize}
\end{definition}

Henceforth, a domain $\mathbb{D}$ is called a \textbf{rich domain}
if it satisfies minimal richness, diversity\textsuperscript{$+$}, and the Interior\textsuperscript{+} and Exterior\textsuperscript{+} properties.
Clearly, the universal domain $\mathbb{P}$ is a rich domain.
In Appendix \ref{app:anexample},
we provide an example of a  domain, and verify that it is a rich domain.
In the Supplementary Material,
we establish two clarifications to show that
the multidimensional hybrid domain and its intersection with the separable domain are both rich domains.\footnote{Since $\mathbb{D}_{\textrm{S}} \cap \mathbb{D}_{\textrm{MH}}(\prec, \underline{x}, \overline{x}) = \mathbb{D}_{\textrm{S}}$ when $\underline{x}^s = \min^{\prec^s}(A^s)$ and $\overline{x}^s = \max^{\prec^s}(A^s)$ hold for all $s \in M$,
this implies that the separable domain $\mathbb{D}_{\textrm{S}}$ is also a rich domain.}

The main theorem below shows that under the richness condition,
multidimensional hybridness is necessary and sufficient for a domain to be a decomposable domain.

\begin{theorem}\label{thm}
Let $\mathbb{D}$ be a rich domain.
Then, $\mathbb{D}$ is a decomposable domain if and only if it is a multidimensional hybrid domain.
\end{theorem}

The proof of the Theorem is contained in Appendix \ref{app:theorem}.

By combining Proposition \ref{prop:FBR} and Theorem \ref{thm},
we obtain the Corollary below that provides a full characterization of strategy-proof rules on a rich multidimensional hybrid domain.

\begin{corollary}\label{cor:characterization}
Let $\mathbb{D}$ be a rich multidimensional hybrid domain on $\prec$ w.r.t. $\underline{x}$ and $\overline{x}$.
An SCF $f: \mathbb{D}^n \rightarrow A$ is a strategy-proof rule if and only if
$f$ is decomposable,
$f^s$ is an FBR for each $s \in M\backslash M(\underline{x}, \overline{x})$, and
$f^t$ is an $(\underline{x}^t, \overline{x}^t)$-FBR for each $t \in M(\underline{x}, \overline{x})$.
\end{corollary}


\section{Final remark and literature review}\label{sec:conclusion}

In a class of rich preference domains, multidimensional hybrid domains are shown to be the unique decomposable domains, which enables us to provide a characterization of strategy-proof rules on these domains.

\citet{BJ1983} initiated the study of strategy-proof SCFs in a multidimensional setting where preferences over the real space $\mathbb{R}^m$ are assumed to be separable and star-shaped (which can be viewed as a variant of single-peakedness over $\mathbb{R}^m$).
They characterized that decomposability is necessary and sufficient for strategy-proofness of a rule, and further justified the salience of separability by showing that their decomposability result degenerates to an impossibility result of the Gibbard-Satterthwaite Theorem as soon as separability is slightly tampered with.
Followed by \citet{BSZ1991}, \citet{BGS1993} and \citet{LW1999}, the preference restrictions of inclusion/exclusion separability, multidimensional single-peakedness and top-separability have been introduced respectively. The characterizations of strategy-proof rules on these restricted domains indicate that all these domains are decomposable domains.\footnote{More discussion on strategy-proof SCFs in multidimensional settings can be found in the two comprehensive survey papers of \citet{S1995} and \citet{B2011}.}
Our class of multidimensional hybrid domains contains these important domains, and our characterization of strategy-proof rules covers their characterization results as special cases.
It is worth mentioning that \citet{BGS1993} and \citet{LW1999} were only able to derive decomposability for \emph{voting schemes}\footnote{A voting scheme is a function $g: \underset{n}{\underbrace{A\times \dots \times A}} \rightarrow A$ that associates to each profile of alternatives, an alternative.
A tops-only SCF degenerates to a voting scheme.} since their domains contain non-separable preferences and they did not derive marginal preferences as we do here. Therefore, their characterizations of strategy-proof rules require a combination of the tops-only property endogenously established on rules, with the decomposition of the corresponding voting schemes. Our way of deriving marginal preferences provides a unified approach for analyzing decomposability of strategy-proof rules in a multidimensional model involving both separable and non-separable preferences.
\citet{LS1999} restricted attention to separable preferences and introduced an elegant richness condition (loosely speaking, sufficiently many \emph{lexicographically separable} preferences are included) on the preference domain that ensures its decomposability.
Our richness condition is different, and mainly related to the no-restoration condition widely explored in the recent literature investigating the equivalence between strategy-proofness and local strategy-proofness  \citep[e.g.,][]{S2013,KRSYZ2021a}.
More importantly, our theorem not only shows that multidimensional hybridness under our richness condition is sufficient for the domain to be a decomposable domain, but also characterizes its necessity.
Recently, \citet{GMS2020} study a multidimensional model where the preference over $\mathbb{R}^m$ is measured by a norm towards the preference peak;
their main result shows that \emph{the marginal median mechanism} (i.e., a combination of median marginal rules at all components) is strategy-proof if and only if the norm satisfies \emph{orthant monotonicity} - a condition that implies the star-shape preference restriction of \citet{BJ1983} and generalizes the requirement of separability.
\citet{CZ2019} also investigate the domain implication of strategy-proof rules in a multidimensional setting, and have characterized that under a mild richness condition, multidimensional single-peakedness is necessary and sufficient for the existence of an anonymous and strategy-proof rule.
Their investigation however cannot be used to detect decomposability of all strategy-proof rules.
Our paper does not concentrate on specific SCFs like the marginal median mechanism, or exogenously require the SCF to be anonymous, but focuses on an environment that ensures decomposability of all strategy-proof rules.

\setlength{\bibsep}{0ex}

\newpage

\appendix

\section*{Appendix}

\section{An example of a rich domain}\label{app:anexample}

In this section, we specify an example of a multidimensional domain, and mainly verify the Interior\textsuperscript{+} and Exterior\textsuperscript{+} properties on the domain.

\begin{figure}[t]
\begin{center}
\includegraphics[width=0.35\textwidth]{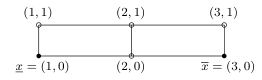}\\[-1em]
\caption{The Cartesian product of linear orders $\prec= \prec^1 \times \prec^2$}\label{fig:grid2}
\end{center}
\end{figure}

\begin{table}[t]

\vspace{-1em}
\hspace{-4em}
{\scriptsize
\begin{tabular}{cccccccccccccccc}
  $P_1$   & $P_2$   & $P_3$   & $P_4$   & $P_5$   & $P_6$   & $P_7$   & $P_8$   & $P_9$   & $P_{10}$& $P_{11}$& $P_{12}$& $P_{13}$& $P_{14}$\\
  $(1,0)$ & $(1,0)$ & $(1,0)$ & $(1,0)$ & $(1,0)$ & $(1,0)$ & $(2,0)$ & $(2,0)$ & $(2,0)$ & $(2,0)$ & $(3,0)$ & $(3,0)$ & $(3,0)$ & $(3,0)$ \\[-0.3em]
  $(2,0)$ & $(3,0)$ & $(3,0)$ & $(3,0)$ & $(1,1)$ & $(1,1)$ & $(1,0)$ & $(1,0)$ & $(2,1)$ & $(2,1)$ & $(1,0)$ & $(1,0)$ & $(3,1)$ & $(3,1)$ \\[-0.3em]
  $(3,0)$ & $(2,0)$ & $(2,0)$ & $(1,1)$ & $(3,0)$ & $(3,0)$ & $(3,0)$ & $(2,1)$ & $(1,0)$ & $(1,0)$ & $(2,0)$ & $(3,1)$ & $(1,0)$ & $(1,0)$ \\[-0.3em]
  $(1,1)$ & $(1,1)$ & $(1,1)$ & $(2,0)$ & $(2,0)$ & $(3,1)$ & $(2,1)$ & $(3,0)$ & $(3,0)$ & $(1,1)$ & $(3,1)$ & $(2,0)$ & $(2,0)$ & $(1,1)$ \\[-0.3em]
  $(2,1)$ & $(2,1)$ & $(3,1)$ & $(3,1)$ & $(3,1)$ & $(2,0)$ & $(1,1)$ & $(1,1)$ & $(1,1)$ & $(3,0)$ & $(1,1)$ & $(1,1)$ & $(1,1)$ & $(2,0)$ \\[-0.3em]
  $(3,1)$ & $(3,1)$ & $(2,1)$ & $(2,1)$ & $(2,1)$ & $(2,1)$ & $(3,1)$ & $(3,1)$ & $(3,1)$ & $(3,1)$ & $(2,1)$ & $(2,1)$ & $(2,1)$ & $(2,1)$ \\[-0.3em]
  &&&&&&&&&&&&&\\[-1em]
  $P_{15}$& $P_{16}$& $P_{17}$& $P_{18}$& $P_{19}$& $P_{20}$& $P_{21}$& $P_{22}$& $P_{23}$& $P_{24}$& $P_{25}$& $P_{26}$& $P_{27}$& $P_{28}$& $P_{29}$& $P_{30}$ \\
  $(1,1)$ & $(1,1)$ & $(1,1)$ & $(1,1)$ & $(2,1)$ & $(2,1)$ & $(2,1)$ & $(2,1)$ & $(2,1)$ & $(2,1)$ & $(3,1)$ & $(3,1)$ & $(3,1)$ & $(3,1)$ & $(3,1)$ & $(3,1)$  \\[-0.3em]
  $(1,0)$ & $(1,0)$ & $(3,1)$ & $(3,1)$ & $(2,0)$ & $(2,0)$ & $(2,0)$ & $(2,0)$ & $(3,1)$ & $(3,1)$ & $(3,0)$ & $(3,0)$ & $(1,1)$ & $(1,1)$ & $(1,1)$ & $(2,1)$  \\[-0.3em]
  $(3,1)$ & $(3,1)$ & $(1,0)$ & $(2,1)$ & $(1,1)$ & $(1,1)$ & $(3,1)$ & $(3,1)$ & $(2,0)$ & $(1,1)$ & $(1,1)$ & $(1,1)$ & $(3,0)$ & $(2,1)$ & $(2,1)$ & $(1,1)$  \\[-0.3em]
  $(3,0)$ & $(2,1)$ & $(2,1)$ & $(1,0)$ & $(1,0)$ & $(3,1)$ & $(1,1)$ & $(1,1)$ & $(1,1)$ & $(2,0)$ & $(1,0)$ & $(2,1)$ & $(2,1)$ & $(3,0)$ & $(3,0)$ & $(3,0)$  \\[-0.3em]
  $(2,1)$ & $(3,0)$ & $(3,0)$ & $(3,0)$ & $(3,1)$ & $(1,0)$ & $(1,0)$ & $(3,0)$ & $(3,0)$ & $(3,0)$ & $(2,1)$ & $(1,0)$ & $(1,0)$ & $(1,0)$ & $(2,0)$ & $(2,0)$  \\[-0.3em]
  $(2,0)$ & $(2,0)$ & $(2,0)$ & $(2,0)$ & $(3,0)$ & $(3,0)$ & $(3,0)$ & $(1,0)$ & $(1,0)$ & $(1,0)$ & $(2,0)$ & $(2,0)$ & $(2,0)$ & $(2,0)$ & $(1,0)$ & $(1,0)$
\end{tabular}
\caption{Domain $\mathbb{D}$}\label{tab:AMH}
}
\end{table}

Let $A = A^1 \times A^2$ where $A^1 = \{1,2,3\}$ and $A^2 = \{0,1\}$ are respectively endowed with
the natural linear orders $\prec^1$ and $\prec^2$.
The grid $\prec = \prec^1 \times \prec^2$ is specified in Figure \ref{fig:grid2}.
We fix two thresholds $\underline{x} = (1,0)$ and $\overline{x} = (3,0)$.
A domain $\mathbb{D}$ of 30 multidimensional hybrid preferences on $\prec$ w.r.t.~$\underline{x}$ and $\overline{x}$ is specified in Table \ref{tab:AMH}. It is evident that $\mathbb{D}$ is minimally rich and satisfies diversity\textsuperscript{+} (see $P_1$ and $P_{30}$).
Note that $\mathbb{D}$ contains both separable preferences (e.g., $P_1$) and non-separable preferences (e.g., $P_2$).
Indeed, $\mathbb{D}$ does not contain all multidimensional hybrid preferences on $\prec$ w.r.t.~$\underline{x}$ and $\overline{x}$, e.g.,
$P_i =  (1,1)_{\rightharpoonup}(2,1)_{\rightharpoonup}(3,1)_{\rightharpoonup}(1,0)_{\rightharpoonup}(2,0)_{\rightharpoonup}(3,0)$
is multidimensional hybrid on $\prec$ w.r.t.~$\underline{x}$ and $\overline{x}$,
but is not included in $\mathbb{D}$.
It is true that $\mathbb{D}$ is a multidimensional hybrid domain:
(i) all preferences of $\mathbb{D}$ are multidimensional hybrid on $\prec$ w.r.t.~$\underline{x}$ and $\overline{x}$, and
(ii) for each $s\in M$, the induced marginal domain $[\mathbb{D}]^s=\mathbb{P}^s$.

\begin{figure}[t]

\hspace{-5.5em}
  \includegraphics[width=1.3\textwidth]{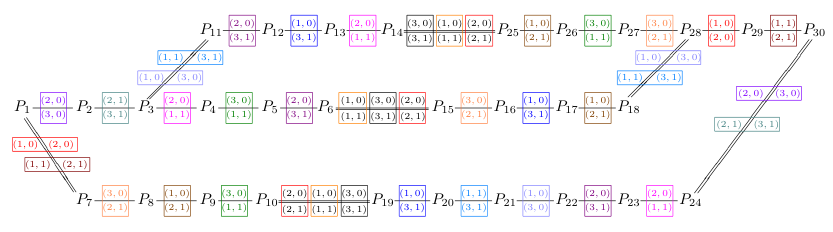}\\

\vspace{-2em}
\caption{The graph $G_{\sim/\sim^+}^{\mathbb{D}}$\protect\footnotemark}\label{fig:adjacency+}

\end{figure}

\footnotetext{In the graph $G_{\sim/\sim^+}^{\mathbb{D}}$,
for instance, the symbol ``$P_1 \frac{\overline{(2,0)}}{~~\underline{(3,0)}~~} P_2 \hspace{-1.25cm}\rule[-2.2mm]{0.11mm}{0.62cm} \hspace{0.57cm}\rule[-2.2mm]{0.11mm}{0.62cm}$~~~~~\;''
represents that\vspace{0.4em}
$P_1 \sim P_2$, $(2,0)\mathrel{P_1}(3,0)$ and $(3,0)\mathrel{P_2}(2,0)$,
while the symbol ``$P_6 \frac{\underline{~~\overline{(1,0)}~\overline{(3,0)}~\overline{(2,0)}~~}}{~~\underline{(1,1)}~\underline{(3,1)}~\underline{(2,1)}~~} P_{15} \hspace{-2.785cm}\rule[-2.2mm]{0.11mm}{0.67cm} \hspace{0.57cm}\rule[-2.2mm]{0.11mm}{0.67cm}
\hspace{0.11cm}\rule[-2.2mm]{0.11mm}{0.67cm} \hspace{0.56cm}\rule[-2.2mm]{0.11mm}{0.67cm}
\hspace{0.12cm}\rule[-2.2mm]{0.11mm}{0.67cm} \hspace{0.565cm}\rule[-2.2mm]{0.11mm}{0.67cm}
$~~~~~~\;''\vspace{0.4em}
represents that $P_6 \sim^+ P_{15}$,
$(1,0)\mathrel{P_6}(1,1)$, $(1,1)\mathrel{P_{15}}(1,0)$,
$(3,0)\mathrel{P_6}(3,1)$, $(3,1)\mathrel{P_{15}}(3,0)$,
$(2,0)\mathrel{P_6}(2,1)$ and $(2,1)\mathrel{P_{15}}(2,0)$.
}

To verify the Interior\textsuperscript{+} and Exterior\textsuperscript{+} properties, we specify
the graph $G_{\sim/\sim^+}^{\mathbb{D}}$ in Figure \ref{fig:adjacency+}.
The following six paths in $G_{\sim/\sim^+}^{\mathbb{D}}$
indicate the Interior\textsuperscript{+} property:
\begin{itemize}
\item $(P_1, P_2,P_3,P_4,P_5,P_6)$, where each preference has the peak $(1,0)$,
\item $(P_7, P_8,P_9,P_{10})$, where each preference has the peak $(2,0)$,
\item $(P_{11}, P_{12},P_{13},P_{14})$, where each preference has the peak $(3,0)$,
\item $(P_{15}, P_{16},P_{17},P_{18})$, where each preference has the peak $(1,1)$,
\item $(P_{19}, P_{20},P_{21},P_{22},P_{23},P_{24})$, where each preference has the peak $(2,1)$, and
\item $(P_{25}, P_{26},P_{27},P_{28}, P_{29}, P_{30})$, where each preference has the peak $(3,1)$.
\end{itemize}

Next, we turn to the no-detour condition in the Exterior\textsuperscript{+} property.
Indeed, it suffices to show that all preferences that have similar peaks form a connected subgraph in $G_{\sim/\sim^+}^{\mathbb{D}}$.
First, note that the edge between $P_3$ and $P_{11}$ and the edge between $P_1$ and $P_7$ combine
the three paths $(P_1, P_2,P_3,P_4,P_5,P_6)$,
$(P_7, P_8,P_9,P_{10})$ and $(P_{11}, P_{12},P_{13},P_{14})$,
and all these preferences have peaks in $(A^1, 0)=\{(1,0),(2,0), (3,0)\}$.
Similarly, the edge between $P_{18}$ and $P_{28}$ and the edge between $P_{24}$ and $P_{30}$ combine the three
paths $(P_{15}, P_{16},P_{17},P_{18})$,
$(P_{19}, P_{20},P_{21},P_{22},P_{23},P_{24})$ and $(P_{25}, P_{26},P_{27},P_{28}, P_{29}, P_{30})$,
and all these preferences have peaks in $(A^1, 1)=\{(1,1),(2,1), (3,1)\}$.
Next, note that
(i) all preferences in the path $(P_1, P_2,P_3,P_4,P_5,P_6, P_{15}, P_{16},P_{17},P_{18})$ have peaks in $(1, A^2) = \{(1,0),(1,1)\}$,
(ii) all preferences in the path $(P_7, P_8,P_9,P_{10}, P_{19}, P_{20},P_{21},P_{22},P_{23},P_{24})$
have peaks in $(2, A^2) = \{(2,0),(2,1)\}$, and
(iii) all preferences in the path
$(P_{11}, P_{12},P_{13},P_{14}, P_{25}, P_{26},P_{27},P_{28}, P_{29}, P_{30})$
have peaks in $(3, A^2)=\{(3,0),(3,1)\}$.
Hence, the no-detour condition is satisfied.

Last, we verify the first condition of the Exterior\textsuperscript{+} property.
We first make an important observation on $G_{\sim/\sim^+}^{\mathbb{D}}$:
each pair of distinct preferences is connected by three distinct paths, and
whenever a restoration on a path connecting two preferences is spotted,
one can immediately identify another path between these two preferences that has no such a restoration.
For instance, the clockwise path
$(P_1, P_2,P_3,P_4,P_5,P_6,P_{15}, P_{16},P_{17},P_{18},P_{28}, P_{29}, P_{30},P_{24},P_{23})$
has $\{(2,0), (1,1)\}$-restoration, i.e., $(2,0)\mathrel{P_3}(1,1)$, $(1,1)\mathrel{P_4}(2,0)$ and $(2,0)\mathrel{P_{23}}(1,1)$,
while $(2,0)$ and $(1,1)$ have not been locally switched on the counter-clockwise path $(P_1, P_7,P_8,P_9,P_{10},P_{19}, P_{20},P_{21},P_{22},P_{23})$.
Therefore, we can conclude that
given arbitrary $P_i, P_i' \in \mathbb{D}$ such that $r_1(P_i) \neq r_1(P_i')$ and distinct $a,b\in A$,
there exists a path in $G_{\sim/\sim^+}^{\mathbb{D}}$ connecting $P_i$ and $P_i'$ that has no $\{a,b\}$-restoration.
In conclusion, $\mathbb{D}$ is a rich domain.

\section{Proof of Theorem \ref{thm}}\label{app:theorem}

We first provide a sketch of the proof, which consists of three parts below.

\noindent
\textbf{Part 1.}
We explore the rich domain $\mathbb{D}$, and establish three results that will be applied for the proof of both the sufficiency and necessity parts of the Theorem.
First, we adopt Proposition 2 of \citet{CZ2019} to show that every strategy-proof rule on $\mathbb{D}$ satisfies the tops-only property (see Lemma \ref{lem:tops-onlydomain}).
Second, for each $s \in M$ and $x^{-s} \in A^{-s}$,
we induce a connected graph over the alternatives of $(A^s, x^{-s})$, as an implication the no-detour condition of the Exterior\textsuperscript{+} property (see Lemma \ref{lem:path-connectedness}).
Last, referring to these induced graphs,
we partially characterize strategy-proof rules on $\mathbb{D}$ (see Lemmas \ref{lem:step1} and \ref{lem:step2}).

\noindent
\textbf{Part 2.}
We prove the sufficiency part of the Theorem: a rich multidimensional hybrid domain $\mathbb{D}$ is a decomposable domain.
First,
by Lemmas \ref{lem:tops-onlydomain}, \ref{lem:path-connectedness} and \ref{lem:step2}
established in \textbf{Part 1},
we show that every strategy-proof rule on $\mathbb{D}$ is decomposable, and all marginal SCFs are strategy-proof marginal rules (see Lemma \ref{lem:decomposition}).
Conversely,
we show that an SCF on $\mathbb{D}$ assembled by strategy-proof marginal rules,
which all are indeed FBRs by the necessity part of Proposition \ref{prop:FBR}, is a strategy-proof rule (see Lemma \ref{lem:assembling}).

\noindent
\textbf{Part 3.}~
We verify the necessity part of the Theorem:
if a rich domain $\mathbb{D}$ is a decomposable domain, it is a multidimensional hybrid domain.
The verification consists of three steps.
In the first step, for each $s \in M$,
we first, by Lemmas \ref{lem:tops-onlydomain} and \ref{lem:step2} and the decomposable-domain hypothesis,
show that every strategy-proof marginal rule on $[\mathbb{D}]^s$ satisfies the tops-only property (see Lemma \ref{lem:tops-onlymarginaldomain}), and
next adopt Lemma \ref{lem:path-connectedness} to show that $G_{\approx}^{A^s}$ is a connected graph (see Lemma \ref{lem:connectedgraph}).
This immediately allows us to reveal by Corollary 2 of \citet{CZ2023} that all induced marginal preferences of $[\mathbb{D}]^s$ are hybrid on $\prec^s$ w.r.t.~some marginal thresholds $\underline{x}^s$ and $\overline{x}^s$.
In the second step,
we by the sufficiency part of Proposition \ref{prop:FBR}, fix several 3-voter strategy-proof FBRs on the induced marginal domains, and then by applying the decomposable-domain hypothesis, assemble them to construct two strategy-proof rules on the domain $\mathbb{D}$.
We then show that as an implication of strategy-proofness these two rules,
all preferences of $\mathbb{D}$ are multidimensional hybrid on $\prec$ w.r.t.~the thresholds $\underline{x}=(\underline{x}^1, \dots, \underline{x}^m)$ and $\overline{x} =(\overline{x}^1, \dots, \overline{x}^m)$ (see Lemma \ref{lem:MH}).
This hence meets condition (i) of Definition \ref{def:AMH}.
In the last step, we establish two lemmas to meet condition (ii) of Definition \ref{def:AMH}.
We first identify a condition implied by the first condition of the Exterior\textsuperscript{+} property
in each induced marginal domain (see Lemma \ref{lem:implication}).
Then, using this condition,
we show that when the two marginal thresholds $\underline{x}^s$ and $\overline{x}^s$ identified in the first step are distinct,
the subgraph $G_{\approx}^{\langle \underline{x}^s,\,\overline{x}^s\rangle}$ has no leaf (see Lemma \ref{lem:noleaf}).
This hence concludes the whole proof.

\medskip

Now, we start the proof of \textbf{Part 1}.
Let $\mathbb{D}$ be a rich domain investigated in both the sufficiency and necessity parts of the Theorem.

\begin{lemma}\label{lem:tops-onlydomain}
For all $n \geq 2$, every strategy-proof rule $f: \mathbb{D}^n \rightarrow A$ satisfies the tops-only property.
\end{lemma}
\begin{proof}
To show the Lemma, we adopt Proposition 2 of \citet{CZ2019},
which says that on a domain of top-separable preferences
satisfying the Interior\textsuperscript{+} and Exterior\textsuperscript{+} properties,
every strategy-proof rule satisfies the tops-only property.
Hence, by the richness of $\mathbb{D}$, to complete the verification, it suffices to show that all preferences are top-separable.

First, let $\mathbb{D}$ be the domain investigated in the sufficiency part of the Theorem.
Thus, $\mathbb{D}$ is a multidimensional hybrid domain on $\prec$ w.r.t.~some thresholds $\underline{x}$ and $\overline{x}$.
Immediately, by condition (i) of Definition \ref{def:MH}, we know that all preferences of $\mathbb{D}$ are top-separable, as required.

Next, let $\mathbb{D}$ be a rich domain investigated in the necessity part of the Theorem.
Thus, $\mathbb{D}$ is a decomposable domain.
Given arbitrary $P_i \in \mathbb{D}$, say $r_1(P_i) = x$, and similar $a,b \in A$, say $M(a,b) = \{s\}$, let $a^s = x^s$. We show $a\mathrel{P_i}b$ via strategy-proofness of some constructed SCF.
Let $N = \{1,2\}$.
At the component $s$, we construct a dictatorial marginal SCF $f^s: \big[[\mathbb{D}]^s\big]^2 \rightarrow A^s$ where voter $1$ is the dictator
(i.e., $f^s([P_1]^s, [P_2]^s) = r_1([P_1]^s)$ for all $[P_1]^s, [P_2]^s \in [\mathbb{D}]^s$),
while at each $t \in M\backslash \{s\}$, we construct a dictatorial marginal SCF $f^t: \big[[\mathbb{D}]^t\big]^2 \rightarrow A^t$ where voter $2$ is the dictator.
Clearly, every marginal SCF here is a strategy-proof marginal rule.
Then, by the decomposable-domain hypothesis, we assemble a strategy-proof rule $f: \mathbb{D}^2 \rightarrow A$ such that
$f(P_1, P_2) = \big(f^1([P_1]^1, [P_2]^1), \dots, f^m([P_1]^m, [P_2]^m)\big)$ for all $P \in \mathbb{D}^n$.
Now, given $P_1,P_1',P_2 \in \mathbb{D}$ such that $P_1 = P_i$, $P_2 \in \mathbb{D}^b$ and $P_1'=P_2$,
we have $f(P_1,P_2) = (x^s, b^{-s}) = a$ and $f(P_1',P_2) = b$,
which imply $a\mathrel{P_1}b$ by strategy-proofness, as required.
Hence, all preferences of $\mathbb{D}$ are top-separable.
This proves the Lemma.
\end{proof}

For the next lemma, we introduce the notion of strong connectedness\textsuperscript{+} between alternatives. Formally, two alternatives $a$ and $b$ are said \textbf{strongly connected\textsuperscript{+}}, denoted $a \approx^+ b$,
if there exist $P_i,P_i' \in \mathbb{D}$ such that
$r_1(P_i) = a$, $r_1(P_i') =b$ and
$P_i \sim^+ P_i'$.
Accordingly, given a nonempty subset $B \subseteq A$,
we can induce a graph $G_{\approx^+}^B \coloneqq \langle B, \mathcal{E}_{\approx^+}^B\rangle$
such that two alternatives of $B$ form an edge if and only if they are strongly connected\textsuperscript{+}.
The lemma below explores the subgraph $G_{\approx^+}^{(A^s,\, x^{-s})}$ for each $s \in M$ and
$x^{-s} \in A^{-s}$.
For instance, recalling the domain of Appendix \ref{app:anexample},
the graph $G_{\approx^+}^A$ is specified in Figure \ref{fig:strongconnectedness+}.

\begin{figure}[t]
\begin{center}
\includegraphics[width=0.5\textwidth]{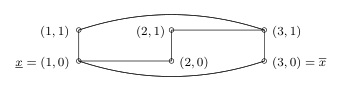}\\[-1em]
\caption{The graph $G_{\approx^+}^A$}\label{fig:strongconnectedness+}
\end{center}
\end{figure}

\begin{lemma}\label{lem:path-connectedness}
Given $s \in M$ and $x^{-s} \in A^{-s}$,
$G_{\approx^+}^{(A^s,\, x^{-s})}$ is a connected graph.
%
\end{lemma}

\begin{proof}
Given distinct $a, b\in (A^s, x^{-s})$, we construct a path in $G_{\approx^+}^{(A^s,\, x^{-s})}$ connecting $a$ and $b$.
Fixing $P_i, P_i' \in \mathbb{D}$ such that $r_1(P_i)=a$ and $r_1(P_i')=b$ by minimal richness,
by the no-detour condition, we have a path $\pi =  (P_{i|1}, \dots, P_{i|v})$ in $G_{\sim/\sim^+}^{\mathbb{D}}$ connecting $P_i$ and $P_i'$ such that
$r_1(P_{i|k}) \in (A^s, x^{-s})$ for all $k = 1, \dots, v$.
Furthermore, we partition the path $\pi$ according to preference peaks (without changing the orders of preferences in $\pi$):
\begin{align*}
\left(\frac{~P_{i|1}, \dots, P_{i|k_1}~}{\textrm{peak $x_1$}}, \dots,
\frac{~P_{i|k_{p-1}+1}, \dots, P_{i|k_p}~}{\textrm{peak $x_p$}},
\frac{~P_{i|k_p+1}, \dots, P_{i|k_{p+1}}~}{\textrm{peak $x_{p+1}$}},\dots,
\frac{~P_{i|k_{q-1}+1}, \dots, P_{i|k_q}~}{\textrm{peak $x_q$}}\right),
\end{align*}
where $q \geq 2$, $k_0 = 0$, $k_q = v$,
$r_1(P_{i|k_{p-1}+1}) = \dots = r_1(P_{i|k_p}) = x_p$ for all $p =1, \dots, q$, and
$x_p \neq x_{p+1}$ for all $p=1, \dots, q-1$.
Thus, we have a sequence of alternatives $(x_1, \dots, x_q)$.
\medskip

\noindent
\textsc{Claim 1}:
We have $x_p \approx^+ x_{p+1}$ for all $1\leq p < q$.\medskip

Given $1\leq p < q$, we have the preferences $P_{i|k_p}$ and $P_{i|k_p+1}$ and the peaks $r_1(P_{i|k_p}) = x_p$ and $r_1(P_{i|k_p+1}) = x_{p+1}$.
Since $x_p, x_{p+1} \in (A^s, x^{-s})$, we write $x_p = (x^s, x^{-s})$ and $x_{p+1} = (y^s, x^{-s})$ where $x^s \neq y^s$.
Clearly, either $P_{i|k_p} \sim P_{i|k_p+1}$ or $P_{i|k_p} \sim^+ P_{i|k_p+1}$ holds.
If $P_{i|k_p} \sim^+ P_{i|k_p+1}$, it is evident that $x_p \approx^+ x_{p+1}$.
We complete the verification by ruling out the possibility that $P_{i|k_p} \sim P_{i|k_p+1}$.
Suppose by contradiction that $P_{i|k_p} \sim P_{i|k_p+1}$.
Thus, $(x^s, x^{-s})$ and $(y^s, x^{-s})$ are the unique two alternatives that are oppositely ranked across $P_{i|k_p}$ and $P_{i|k_p+1}$.
However, since $P_{i|k_p}$ and $P_{i|k_p+1}$ are shown to be top-separable in the proof of Lemma \ref{lem:tops-onlydomain}, we have
$(x^s, z^{-s})\mathrel{P_{i|k_p}}(y^s, z^{-s})$ and $(y^s, z^{-s})\mathrel{P_{i|k_p+1}}(x^s, z^{-s})$ for all $z^{-s} \in A^{-s}$ - a contradiction.
This completes the verification of the claim.\medskip

Note that some alternatives may appear multiple times in the sequence $(x_1, \dots, x_q)$.
For instance, let $x_p = x_{p'}$ where $1 \leq p < p' \leq q$.
Since $x_p\neq x_{p+1}$, it is clear that $p+1< p'$.
We then remove alternatives $x_{p+1}, \dots, x_{p'}$, and refine the sequence
to $(x_1, \dots, x_p, x_{p'+1}, \dots, x_q)$,
where any consecutive alternatives remain to be strongly connected\textsuperscript{+}.
By repeatedly eliminating repetitions of alternatives, we finally construct a path in $G_{\approx^+}^{(A^s,\, x^{-s})}$ connecting $a$ and $b$.
%
\end{proof}

Furthermore, we fix an arbitrary strategy-proof rule $f: \mathbb{D}^n \rightarrow A$, and partially characterize $f$ in the next two lemmas using the subgraphs $G_{\approx^+}^{(A^s,\,x^{-s})}$ for all $s\in M$ and $x^{-s} \in A^{-s}$.
Clearly, by Lemma \ref{lem:tops-onlydomain}, $f$ satisfies the tops-only property.
For notational convenience, we henceforth write $(a, P_{-i})$ to denote a preference profile where voter $i$ reports an arbitrary preference with the peak $a$, and all others report $P_1, \dots, P_{i-1}, P_{i+1}, \dots, P_n$ respectively.
For ease of presentation, given $P_i \in \mathbb{D}$ such that $r_1(P_i) = a$,
let $r_1(P_i)^s \coloneqq a^s$ for all $s \in M$;
given $P \in \mathbb{D}$ and $f(P) = a$, let $f(P)^s \coloneqq a^s$ for all $s \in M$.


\begin{lemma}\label{lem:step1}
Given $i \in N$, $P_i, P_i' \in \mathbb{D}$ and $P_{-i} \in \mathbb{D}^{n-1}$,
let $M\big(r_1(P_i), r_1(P_i')\big) = \{s\}$.
We have $f(P_i, P_{-i})^t = f(P_i', P_{-i})^t$ for all $t \in M\backslash \{s\}$.
\end{lemma}

\begin{proof}
For notational convenience,
let $r_1(P_i) = (a^s, z^{-s})$ and $r_1(P_i') = (b^s, z^{-s})$,
where $a^s \neq b^s$.
By Lemma \ref{lem:path-connectedness},
we have a path $(a_1, \dots, a_v)$
in $G_{\approx^+}^{(A^s,\, z^{-s})}$ connecting $(a^s, z^{-s})$ and $(b^s, z^{-s})$.
For ease of presentation, let $f(a_k, P_{-i}) = x_k$ for all $k = 1, \dots, v$.
Clearly, $f(P_i, P_{-i}) = f(a_1, P_{-i}) = x_1$ and $f(P_i', P_{-i})=f(a_v, P_{-i})=x_v$.
To complete the verification, it suffices to show $x_k^{-s} = x_{k+1}^{-s}$ for all $k = 1, \dots, v-1$.

Given $1 \leq k < v$, we have $f(a_k, P_{-i}) = x_k$ and $f(a_{k+1}, P_{-i}) = x_{k+1}$.
The result holds evidently if $x_k = x_{k+1}$.
Next, assume $x_k \neq x_{k+1}$.
Since $a_k \approx^+ a_{k+1}$, there exist $\hat{P}_i,\hat{P}_i' \in \mathbb{D}$ such that
$r_1(\hat{P}_i) = a_k$, $r_1(\hat{P}_i') = a_{k+1}$ and
$\hat{P}_i \sim^{+} \hat{P}_i'$.
Since $f(\hat{P}_i, P_{-i}) = f(a_k, P_{-i}) = x_k$ and
$f(\hat{P}_i', P_{-i}) = f(a_{k+1}, P_{-i}) = x_{k+1}$,
strategy-proofness implies $x_k\mathrel{\hat{P}_i}x_{k+1}$ and $x_{k+1}\mathrel{\hat{P}_i'}x_k$.
Since $\hat{P}_i \sim^{+} \hat{P}_i'$, note that any two alternatives that are oppositely ranked across $\hat{P}_i$ and $\hat{P}_i'$, agree on all components other than $s$.
Hence, $x_k^{-s} = x_{k+1}^{-s}$, as required.
\end{proof}

\begin{lemma}\label{lem:step2}
Given $i \in N$, $P_i, P_i' \in \mathbb{D}$, $P_{-i} \in \mathbb{D}^{n-1}$ and $t \in M$,
let $r_1(P_i)^t= r_1(P_i')^t$.
We have $f(P_i, P_{-i})^t = f(P_i', P_{-i})^t$.
\end{lemma}

\begin{proof}
The Lemma immediately follows from the tops-only property if $r_1(P_i) = r_1(P_i')$.
Henceforth, let $r_1(P_i) \neq r_1(P_i')$.
Assume w.l.o.g.~that $M\big(r_1(P_i), r_1(P_i')\big) = \{1, \dots, s\}$, where $1\leq s<m$.
We write $r_1(P_i) = (a^1, \dots, a^s, z^{\{s+1, \dots, m\}})$ and
$r_1(P_i') = (b^1, \dots, b^s, z^{\{s+1, \dots, m\}})$ where $a^k \neq b^k$ for all $k=1, \dots, s$.
Clearly, $r_1(P_i)^t= r_1(P_i')^t$ implies $s< t \leq m$.
We construct the alternative $x_{k} = (b^1, \dots, b^k, a^{k+1}, \dots, a^s, z^{\{s+1, \dots, m\}})$ for each $k = 0, 1, \dots, s$.
Thus, $x_0 = r_1(P_i)$ and
$x_s = r_1(P_i')$.
By minimal richness, for each $k = 0, 1, \dots, s$,
we fix a preference $P_{i|k} \in \mathbb{D}$ such that $r_1(P_{i|k}) = x_k$.
For each $k=0,1, \dots, s-1$, since $t \notin M\big(r_1(P_i^k),r_1(P_i^{k+1})\big)$,
Lemma \ref{lem:step1} implies $f(P_{i|k}, P_{-i})^t = f(P_{i|k+1}, P_{-i})^t$.
Therefore, we have $f(P_i, P_{-i})^t = f(P_{i|0}, P_{-i})^t =\dots  =f(P_{i|s}, P_{-i})^t = f(P_i', P_{-i})^t$.
\end{proof}

This completes the verification in \textbf{Part 1}, and we then turn to \textbf{Part 2}.\medskip

\noindent
\textbf{(Sufficiency Part)}~
Let $\mathbb{D}$ be a rich multidimensional domain on $\prec$ w.r.t.~thresholds $\underline{x}$ and $\overline{x}$.
We show that $\mathbb{D}$ is a decomposable domain.
We fix a strategy-proof rule $f: \mathbb{D}^n \rightarrow A$, and show that $f$ is decomposable, and all marginal SCFs are strategy-proof marginal rules.
Of course, $f$ satisfies the tops-only property by Lemma \ref{lem:tops-onlydomain} and triggers Lemma \ref{lem:step2}.

\begin{lemma}\label{lem:decomposition}
SCF $f$ is decomposable, and each marginal SCF is a strategy-proof marginal rule.
\end{lemma}

\begin{proof}
The proof consists of three claims.\medskip

\noindent
\textsc{Claim 1}: Given $P, P' \in \mathbb{D}^n$ and $t \in M$,
let $r_1(P_i)^t= r_1(P_i')^t$ for all $i \in N$.
We have $f(P)^t = f(P')^t$.\medskip

We first construct the profile
$P(k) = (P_1', \dots, P_k', P_{k+1}, \dots, P_n)$ for each $k = 0, 1, \dots, n$.
Clearly, $P(0) = P$ and $P(n) = P'$.
For each $i = 1, \dots, n$, note that $P(i-1) = (P_1', \dots, P_{i-1}',P_i, P_{i+1}, \dots, P_n)$ and
$P(i)= (P_1', \dots, P_{i-1}',P_i', P_{i+1}, \dots, P_n)$ agree on preferences of all voters other than $i$, and $r_1(P_i)^t= r_1(P_i')^t$.
Then, Lemma \ref{lem:step2} implies $f(P(i-1))^t = f(P(i))^t$.
Hence,  we have $f(P)^t =f(P(0))^t = \dots =f(P(n))^t = f(P')^t$.
This completes the verification of the claim.\medskip

\noindent
\textsc{Claim 2}: SCF $f$ is decomposable.\medskip

Given $s \in M$,
by Claim 1 and minimal richness of $\mathbb{D}$,
we can construct a function $g^s: \underset{n}{\underbrace{A^s \times \dots \times A^s}} \rightarrow A^s$
such that for all $(x_1^s, \dots, x_n^s) \in \underset{n}{\underbrace{A^s \times \dots \times A^s}}$
and all $(P_1, \dots, P_n) \in \mathbb{D}^n$ with $\big(r_1(P_1)^s, \dots, r_1(P_n)^s\big) = (x_1^s, \dots,  x_n^s)$,
$g^s(x_1^s, \dots, x_n^s) = f(P_1, \dots, P_n)^s$.
Then, we construct a marginal SCF $f^s : \big[[\mathbb{D}]^s\big]^n \rightarrow A^s$ such that
$f^s([P_1]^s, \dots, [P_n]^s) = g^s\big(r_1([P_1]^s), \dots, r_1([P_n]^s)\big)$ for all $([P_1]^s, \dots, [P_n]^s) \in \big[[\mathbb{D}]^s\big]^n$.
By construction,
it is clear that for all $(P_1, \dots, P_n) \in \mathbb{D}^n$, we have
$\big[f(P_1, \dots, P_n) = a\big] \Leftrightarrow \big[f^s([P_1]^s, \dots, [P_n]^s) = a^s\; \textrm{for all}\; s \in M\big]$.
Therefore, $f$ is decomposable.
This completes the verification of the claim.\medskip

\noindent
\textsc{Claim 3}: Given $s \in M$, the marginal SCF $f^s$ constructed in Claim 2 is a strategy-proof marginal rule.\medskip

First, we claim that $f^s$ is unanimous.
Given a profile $([P_1]^s, \dots, [P_n]^s) \in \big[[\mathbb{D}]^s\big]^n$,
let $r_1([P_1]^s) = \dots = r_1([P_n]^s) = a^s$.
We show $f^s([P_1]^s, \dots, [P_n]^s) =a^s$.
Given $z^{-s} \in A^{-s}$, by minimal richness, we have a profile $(P_1', \dots, P_n') \in \mathbb{D}^n$ such that $r_1(P_1') = \dots = r_1(P_n') = (a^s, z^{-s})$.
By unanimity of $f$,
it is clear that $f(P_1', \dots, P_n') = (a^s, z^{-s})$.
Then, by the decomposition of $f$,
we have $f^s([P_1]^s, \dots, [P_n]^s) = f(P_1', \dots, P_n')^s = a^s$, as required.

Last, we show strategy-proofness of $f^s$.
Given $i \in N$, $[P_i]^s ,[P_i']^s \in [\mathbb{D}]^s$ and
$[P_{-i}]^s\coloneqq ([P_1]^s, \dots, [P_{i-1}]^s,[P_{i+1}]^s, \dots, [P_n]^s) \in \big[[\mathbb{D}]^s\big]^{n-1}$,
let $f^s([P_i]^s, [P_{-i}]^s) = a^s$, $f^s([P_i']^s, [P_{-i}]^s) = b^s$ and $a^s \neq b^s$.
We show $a^s\mathrel{[P_i]^s}b^s$.
Since $f^s$ by construction satisfies the tops-only property, $f^s([P_i]^s, [P_{-i}]^s) \neq f^s([P_i']^s, [P_{-i}]^s)$ implies $r_1([P_i]^s) \neq r_1([P_i']^s)$.
For notational convenience, let $r_1(P_i) = (z^s, z^{-s})$ and $r_1([P_i']^s) = \hat{z}^s$.
Thus, to show $a^s\mathrel{[P_i]^s}b^s$, it suffices to show $(a^s, z^{-s})\mathrel{P_i}(b^s, z^{-s})$.
By minimal richness, we fix $\hat{P}_i \in \mathbb{D}$ such that $r_1(\hat{P}_i)=(\hat{z}^s, z^{-s})$, and fix $\hat{P}_j \in \mathbb{D}$ such that $r_1(\hat{P}_j) = \big(r_1([P_j]^s), z^{-s}\big)$ for each $j \in N\backslash \{i\}$.
By the decomposition of $f$ and the tops-only property of $f^s$, we have
$f(P_i, \hat{P}_{-i})^s = f^s([P_i]^s, [\hat{P}_{-i}]^s)
= f^s([P_i]^s, [P_{-i}]^s) = a^s$ and
$f(\hat{P}_i, \hat{P}_{-i})^s = f^s([\hat{P}_i]^s, [\hat{P}_{-i}]^s)
= f^s([P_i']^s, [P_{-i}]^s) = b^s$, while
by the decomposition of $f$ and unanimity of marginal rules,
$f(P_i, \hat{P}_{-i})^t = f^t([P_i]^t, [\hat{P}_{-i}]^t)= z^t$ and
$f(\hat{P}_i, \hat{P}_{-i})^t = f^t([\hat{P}_i]^t, [\hat{P}_{-i}]^t)
= z^t$ for all $t \in M\backslash \{s\}$.
Thus,
$f(P_i, \hat{P}_{-i}) = (a^s, z^{-s})$ and $f(\hat{P}_i, \hat{P}_{-i}) = (b^s, z^{-s})$,
which by strategy-proofness imply $(a^s, z^{-s})\mathrel{P_i}(b^s, z^{-s})$, as required.
This completes the verification of the claim, and hence proves the Lemma.
\end{proof}

\vspace{-1em}
Conversely, we fix a strategy-proof marginal rule $f^s: \big[[\mathbb{D}]^s\big]^n \rightarrow A^s$
for each $s \in M$,
assemble an SCF $f: \mathbb{D}^n \rightarrow A$ such that
$f(P) = \big(f^1([P]^1), \dots, f^m([P]^m)\big)$ for all $P \in \mathbb{D}^n$,
and show that $f$ is a strategy-proof rule.

\begin{lemma}\label{lem:assembling}
The assembled SCF $f$ is a strategy-proof rule.
\end{lemma}

\begin{proof}
Since $f^1, \dots, f^m$ are unanimous,
the assembled SCF $f$ by construction must be unanimous, and hence is a rule.
We focus on showing strategy-proofness of $f$.
Since $\mathbb{D}$ is a multidimensional domain on $\prec$ w.r.t.~the thresholds $\underline{x}$ and $\overline{x}$,
we know $\mathbb{D} \subseteq \mathbb{D}_{\textrm{MH}}(\prec, \underline{x}, \overline{x})$ and hence
$[\mathbb{D}]^s \subseteq [\mathbb{D}_{\textrm{MH}}(\prec, \underline{x}, \overline{x})]^s$ for each $s \in M$.
First, by Proposition \ref{prop:FBR}, the marginal SCF $f^s: \big[[\mathbb{D}]^s\big]^n \rightarrow A^s$ is an FBR for each $s \in M\backslash M(\underline{x}, \overline{x})$, and the marginal SCF $f^t: \big[[\mathbb{D}]^t\big]^n \rightarrow A^t$ is an $(\underline{x}^t, \overline{x}^t)$-FBR for each $t\in M(\underline{x}, \overline{x})$.
Next, for each $s \in M$, adopting the fixed ballots $(b_J^s)_{J \subseteq N}$ of $f^s$,
we construct an FBR $\hat{f}^s: \big[[\mathbb{D}_{\textrm{MH}}(\prec, \underline{x}, \overline{x})]^s\big]^n \rightarrow A^s$, which by Proposition \ref{prop:FBR}
is a strategy-proof marginal rule as well.
Then, we assemble an SCF $\hat{f}: \mathbb{D}_{\textrm{MH}}(\prec, \underline{x}, \overline{x})^n \rightarrow A$ such that
$\hat{f}(P) = \big(\hat{f}^1([P]^1), \dots, \hat{f}^m([P]^m)\big)$ for all $P \in \mathbb{D}_{\textrm{MH}}(\prec, \underline{x}, \overline{x})^n$.
Given $\mathbb{D} \subseteq \mathbb{D}_{\textrm{MH}}(\prec, \underline{x}, \overline{x})$,
since $f^s$ and $\hat{f}^s$ share the same fixed ballots at each component $s \in M$,
it is true that
for all $(P_1, \dots, P_n) \in \mathbb{D}^n$,
\begin{align*}
f(P_1, \dots, P_n)
= &~\big(f^1([P_1]^1, \dots, [P_n]^1), \dots, f^m([P_1]^m, \dots, [P_n]^m)\big) \\
= &~\big(\hat{f}^1([P_1]^1, \dots, [P_n]^1), \dots, \hat{f}^m([P_1]^m, \dots, [P_n]^m)\big)
= \hat{f}(P_1, \dots, P_n).
\end{align*}
Hence, it suffices to show strategy-proofness of $\hat{f}$.

Given $i \in N$, $P_i, P_i' \in \mathbb{D}_{\textrm{MH}}(\prec, \underline{x}, \overline{x})$ and $P_{-i} \in \mathbb{D}_{\textrm{MH}}(\prec, \underline{x}, \overline{x})^{n-1}$,
let $\hat{f}(P_i,P_{-i}) = a$, $\hat{f}(P_i',P_{-i}) = b$ and $a \neq b$.
We show $a\mathrel{P_i}b$.
For notational convenience, let $r_1(P_i) = x$.
By the definition of multidimensional hybridness on $\prec$ w.r.t.~$\underline{x}$ and $\overline{x}$,
to show $a\mathrel{P_i}b$, it suffices to show that for all $s \in M(a,b)$, either $a^s = x^s$, or
$a^s \in \textrm{Int}\langle x^s, b^s\rangle$ and $a^s \notin \textrm{Int}\langle \underline{x}^s, \overline{x}^s\rangle$.
Given $s \in M(a,b)$,
we know either $a^s = x^s$ or $a^s \neq x^s$ holds.
Henceforth, we fix $a^s \neq x^s$, and
show $a^s \in \textrm{Int}\langle x^s, b^s\rangle$ and $a^s \notin \textrm{Int}\langle \underline{x}^s, \overline{x}^s\rangle$.
By the decomposition of $\hat{f}$,
we have $\hat{f}([P_i]^s, [P_{-i}]^s) = a^s$ and $\hat{f}([P_i']^s, [P_{-i}]^s) = b^s$.
Since $\hat{f}^s$, as an FBR, satisfies the tops-only property,
we have $\hat{f}([\hat{P}_i]^s, [P_{-i}]^s) = a^s$ for all $[\hat{P}_i]^s \in [\mathbb{D}_{\textrm{MH}}(\prec, \underline{x}, \overline{x})]^s$ such that $r_1([\hat{P}_i]^s) = r_1([P_i]^s) = x^s$.
Immediately, by strategy-proofness of $\hat{f}^s$, we infer that
for all $[\hat{P}_i]^s \in [\mathbb{D}_{\textrm{MH}}(\prec, \underline{x}, \overline{x})]^s$,
$\big[r_1([\hat{P}_i]^s) = x^s\big] \Rightarrow \big[a^s\mathrel{[\hat{P}_i]^s} b^s\big]$.
Since $[\mathbb{D}_{\textrm{MH}}(\prec, \underline{x}, \overline{x})]^s$ contains all
hybrid marginal preferences on $\prec^s$ w.r.t.~$\underline{x}^s$ and $\overline{x}^s$,
it is true that if $a^s \notin \textrm{Int}\langle x^s, b^s\rangle$
or $a^s \in \textrm{Int}\langle \underline{x}^s, \overline{x}^s\rangle$,
there exists an induced marginal preference of $[\mathbb{D}_{\textrm{MH}}(\prec, \underline{x}, \overline{x})]^s$ where $x^s$ is top-ranked, and $b^s$ is ranked above $a^s$.
Therefore, it must be the case that $a^s \in \textrm{Int}\langle x^s, b^s\rangle$ and $a^s \notin \textrm{Int}\langle \underline{x}^s, \overline{x}^s\rangle$, as required.
This completes the verification of the Lemma, and hence proves the sufficiency part of the Theorem.
\end{proof}

This completes the verification in \textbf{Part 2}, and we last turn to \textbf{Part 3}.

\medskip

\noindent
\textbf{(Necessity Part)}~
Let a rich domain $\mathbb{D}$ be a decomposable domain.
We identify two thresholds $\underline{x}$ and $\overline{x}$, and show that $\mathbb{D}$ is a multidimensional hybrid domain on $\prec$ w.r.t.~$\underline{x}$ and $\overline{x}$.
By diversity\textsuperscript{+},
we have two separable preferences $\underline{P}_i,\overline{P}_i \in \mathbb{D}$ that are complete reversals.
For each $s \in M$,
by relabeling elements as necessary, we can assume w.l.o.g.~that for all $a^s, b^s \in A^s$, $[a^s \prec^s b^s] \Rightarrow \big[a^s\mathrel{[\underline{P}_i]^s}b^s\; \textrm{and}\; b^s\mathrel{[\overline{P}_i]^s}a^s\big]$.

\begin{lemma}\label{lem:tops-onlymarginaldomain}
Given $s \in M$,
for all $n \geq 2$,
every strategy-proof marginal rule $f^s: \big[[\mathbb{D}]^s\big]^n \rightarrow A^s$ satisfies the tops-only property.
\end{lemma}
\begin{proof}
Suppose not, i.e., we have a non-tops-only and strategy-proof marginal rule $f^s: \big[[\mathbb{D}]^s\big]^n \rightarrow A^s$ for some $n \geq 2$.
There must exist $i \in N$,
$[P_i]^s,[P_i']^s \in [\mathbb{D}]^s$ with $r_1([P_i]^s) = r_1([P_i']^s)$ and
$[P_{-i}]^s \in \big[[\mathbb{D}]^s\big]^{n-1}$ such that
$f^s([P_i]^s, [P_{-i}]^{s}) \neq f^s([P_i']^s, [P_{-i}]^{s})$.
For each $t \in M\backslash \{s\}$, we fix a dictatorial marginal SCF $f^t: \big[[\mathbb{D}]^t\big]^n \rightarrow A^t$.
We then assemble an SCF $f: \mathbb{D}^n \rightarrow A$ such that
$f(P)=\big(f^1([P]^1), \dots, f^m([P]^m)\big)$ for all $P \in \mathbb{D}^n$.
Since all marginal SCFs are strategy-proof marginal rules,
by the decomposable-domain hypothesis, $f$ is a strategy-proof rule.
Immediately, $f$ satisfies the tops-only property by Lemma \ref{lem:tops-onlydomain}, and triggers Lemma \ref{lem:step2}.
By construction, we have $f(P_i, P_{-i})^s
= f^s\big([P_i]^s, [P_{-i}]^{s}\big) \neq f^s\big([P_i']^s, [P_{-i}]^{s}\big) = f(P_i', P_{-i})^s$.
However, by Lemma \ref{lem:step2},
since $r_1(P_i)^s = r_1([P_i]^s) = r_1([P_i']^s) = r_1(P_i')^s$,
we have $f(P_i, P_{-i})^s = f(P_i', P_{-i})^s$ - a contradiction.
\end{proof}

\begin{lemma}\label{lem:connectedgraph}
Given $s \in M$, consider the induced marginal domain $[\mathbb{D}]^s$.
It is true that $G_{\approx}^{A^s}$ is a connected graph.
\end{lemma}

\begin{proof}
Given distinct $a^s, b^s \in A^s$, we construct a path in $G_{\approx}^{A^s}$ connecting $a^s$ and $b^s$.
Fixing arbitrary $x^{-s} \in A^s$, we have the alternatives $(a^s, x^{-s})$ and $(b^s, x^{-s})$.
Since $G_{\approx^+}^{(A^s, \, x^{-s})}$ is a connected graph by Lemma \ref{lem:path-connectedness},
we have a path $(x_1, \dots, x_v)$ in $G_{\approx^+}^{(A^s, \, x^{-s})}$ connecting
$(a^s, x^{-s})$ and $(b^s, x^{-s})$.
Clearly, $x_1^{-s} = \dots = x_v^{-s} = x^{-s}$ and $x_k^s \neq x_{k'}^s$ for all $1 \leq k<k' \leq v$.
For each $1 \leq k < v$, by $x_k \approx^+ x_{k+1}$,
we have $P_i,P_i' \in \mathbb{D}$ such that
$r_1(P_i) = x_k$, $r_1(P_i') = x_{k+1}$ and $P_i \sim^+ P_i'$,
which implies $r_1([P_i]^s) = r_2([P_i']^s) = x_k^s$, $r_1([P_i']^s) = r_2([P_i]^s) = x_{k+1}^s$ and
$r_{\ell}([P_i]^s) = r_{\ell}([P_i']^s)$ for all $\ell \in \{3, \dots, |A^s|\}$.
Hence, we have $x_k^s \approx x_{k+1}^s$. Thus, $(x_1^s, \dots, x_v^s)$ is a path in $G_{\approx}^{A^s}$ connecting $a^s$ and $b^s$, as required.
\end{proof}

Now, given $s \in M$,
$G_{\approx}^{A^s}$ is a connected graph,
$[\mathbb{D}]^s$ contains complete reversals $[\underline{P}_i]^s$ and $[\overline{P}_i]^s$, and $[\mathbb{D}]^s$ is a top-only marginal domain
(i.e., every strategy-proof marginal rule satisfies the tops-only property).
Then, by Corollary 2 of \citet{CZ2023},
we can infer the following two conditions:

\begin{description}
\item[\rm \textbf{Condition (a)}] there exist marginal thresholds $\underline{x}^s, \overline{x}^s \in A^s$ (either identical or not) such that all induced marginal preferences of $[\mathbb{D}]^s$ are hybrid on $\prec^s$ w.r.t.~$\underline{x}^s$ and $\overline{x}^s$, and

\item[\rm \textbf{Condition (b)}] if $\underline{x}^s \neq \overline{x}^s$,
there exist no linear order $\lhd^s$ over $A^s$ and no marginal thresholds $\hat{\underline{x}}^s,\hat{\overline{x}}^s \in A^s$ (either identical or not) such that
all induced marginal preferences of $[\mathbb{D}]^s$ are hybrid on $\lhd^s$ w.r.t.~$\hat{\underline{x}}^s$ and $\hat{\overline{x}}^s$, and
$\big\langle \hat{\underline{x}}^s, \hat{\overline{x}}^s\big\rangle^{\lhd^s} \coloneqq
\big\{x^s \in A^s: \hat{\underline{x}}^s\, \lhd^s\, x^s\, \lhd^s\, \hat{\overline{x}}^s \big\}\cup \big\{\hat{\underline{x}}^s, \hat{\overline{x}}^s\big\} \subset \langle \underline{x}^s, \overline{x}^s\rangle$.\footnote{Corollary 2 of \citet{CZ2023} provides a classification of non-dictatorial and tops-only domains (i.e., a non-dictatorial domain admits a non-dictatorial, strategy-proof rule, while a tops-only domain endogenizes the tops-only property on all strategy-proof rules) under a richness condition called \emph{unidimensionality} (i.e., a unidimensional domain is analogous to a marginal domain $\mathbb{D}^s$ that induces a connected graph $G_{\approx}^{A^s}$, contains complete reversals, and satisfies an additional condition called \emph{leaf symmetry}), according to the existence and non-existence of anonymous and strategy-proof rules.
Leaf symmetry plays a trivial role in the classification - it is only adopted to ensure that all non-dictatorial domains in question are able to be fully characterized by an explicit preference restriction.
If we enlarge the classification to cover both non-dictatorial and dictatorial domains,
the condition of leaf symmetry can be removed without affecting their proof for the classification.
Then, the classification provides us \textbf{Conditions (a)} and \textbf{(b)} above.}
\end{description}

Now, by \textbf{Condition (a)}, we
assemble the thresholds $\underline{x} = (\underline{x}^1, \dots, \underline{x}^m)$ and $\overline{x} = (\overline{x}^1, \dots, \overline{x}^m)$.
We first show that all preferences of $\mathbb{D}$ are multidimensional hybrid on $\prec$ w.r.t.~$\underline{x}$ and $\overline{x}$.

\begin{lemma}\label{lem:MH}
\makebox{All preferences of $\mathbb{D}$ are multidimensional hybrid on $\prec$ w.r.t.~$\underline{x}$ and $\overline{x}$.}
\end{lemma}

\begin{proof}
Fixing a preference $P_i \in \mathbb{D}$ and two similar alternatives $a,b \in A$,
let $r_1(P_i) = x$, $M(a,b) = \{s\}$,
and either $a^s = x^s$, or $a^s \in \textrm{Int}\langle x^s, b^s\rangle$ and $a^s \notin \textrm{Int}\langle \underline{x}^s, \overline{x}^s\rangle$.
We show $a\mathrel{P_i}b$.
For notational convenience, let $a = (a^s, z^{-s})$ and $b=(b^s, z^{-s})$.
First, by the proof of Lemma \ref{lem:tops-onlydomain}, we know that $P_i$ is a top-separable preference.
Hence, if $a^s = x^s$, it is clear that $a\mathrel{P_i}b$, as required.
Henceforth, let $a^s \in \textrm{Int}\langle x^s, b^s\rangle$ and $a^s \notin \textrm{Int}\langle \underline{x}^s, \overline{x}^s\rangle$.
Clearly, there are two cases: (1) $\underline{x}^s = \overline{x}^s$ and (2) $\underline{x}^s \neq \overline{x}^s$.

In case (1), all induced marginal preferences of $[\mathbb{D}]^s$ are single-peaked on $\prec^s$.
Let $N = \{1,2,3\}$. We construct a median marginal rule: for all $([P_1]^s, [P_2]^s, [P_3]^s) \in \big[[\mathbb{D}]^s\big]^3$,
$f^s([P_1]^s, [P_2]^s, [P_3]^s) = \textrm{med}^{\prec^s}\big(r_1([P_1]^s), r_1([P_2]^s), r_1([P_3]^s)\big)$.
It is clear that $f^s$ is a strategy-proof marginal rule.
For each $t \in M\backslash \{s\}$, we construct a dictatorial marginal rule $f^t: \big[[\mathbb{D}]^t\big]^3 \rightarrow A^t$, where voter $2$ is fixed to be the dictator.
It is also true that $f^t$ is a strategy-proof marginal rule.
We assemble the SCF $f: \mathbb{D}^3 \rightarrow A$ such that for all $(P_1, P_2, P_3) \in \mathbb{D}^3$,
$f(P_1, P_2,P_3) = \big(f^1([P_1]^1, [P_2]^1,[P_3]^1), \dots, f^m([P_1]^m, [P_2]^m, [P_3]^m)\big)$.
Clearly, by the decomposable-domain hypothesis, $f$ is a strategy-proof rule.
Now, given $P_1 = P_i$ and $P_1', P_2, P_3 \in \mathbb{D}$ such that $r_1(P_1') = b$, $r_1(P_2) = a$ and $r_1(P_3) =b$,
we have $f(P_1, P_2, P_3)^s = \textrm{med}^{\prec^s}(x^s, a^s, b^s\big) = a^s$,
$f(P_1', P_2, P_3)^s = \textrm{med}^{\prec^s}(b^s, a^s, b^s\big) = b^s$
and $f(P_1, P_2, P_3)^t = f(P_1', P_2, P_3)^t = r_1([P_2]^t) = z^t$ for all $t \in M\backslash \{s\}$.
Hence, $f(P_1, P_2, P_3) = (a^s, z^{-s}) = a$ and $f(P_1', P_2, P_3) = (b^s, z^{-s}) = b$
which by strategy-proofness imply $a\mathrel{P_1}b$, as required.

In case (2), all induced marginal preferences of $[\mathbb{D}]^s$ are hybrid on $\prec^s$ w.r.t.~$\underline{x}^s$ and $\overline{x}^s$.
Let $N = \{1,2,3\}$. We construct an $(\underline{x}^s, \overline{x}^s)$-FBR $f^s: \big[[\mathbb{D}]^s\big]^3 \rightarrow A^s$ such that
the fixed ballots are specified as follows: $b_{\emptyset}^s =b_{\{2\}}^s=b_{\{3\}}^s =\min^{\prec^s}(A^s)$,
$b_{\{2,3\}}^s = \underline{x}^s$, $b_{\{1\}}^s = \overline{x}^s$ and $b_{\{1,2\}}^s = b_{\{1,3\}}^s = b_{\{1,2,3\}}^s =\max^{\prec^s}(A^s)$.
By Proposition \ref{prop:FBR}, $f^s$ is a strategy-proof marginal rule.\medskip

\noindent
\textsc{Claim 1}: Given $[P_1]^s, [P_2]^s, [P_3]^s \in [\mathbb{D}]^s$, we have \vspace{0.5em}
\begin{align*}
\footnotesize{f^s([P_1]^s, [P_2]^s, [P_3]^s)
=\left\{\!\!
\begin{array}{ll}
r_1([P_1]^s) & \textrm{if}\; r_1([P_1]^s) \in \langle \underline{x}^s, \overline{x}^s\rangle,\\
\textrm{med}^{\prec^s} \big(r_1([P_1]^s), r_1([P_2]^s), r_1([P_3]^s), r_1([P_1]^s), \underline{x}^s\big) & \textrm{if}\; r_1([P_1]^s) \prec^s \underline{x}^s,\\
\textrm{med}^{\prec^s} \big(r_1([P_1]^s), r_1([P_2]^s), r_1([P_3]^s), r_1([P_1]^s), \overline{x}^s\big) & \textrm{if}\; \overline{x}^s \prec^s r_1([P_1]^s).
\end{array}
\right.}
\end{align*}
\vspace{0.1em}

According to the fixed ballots $\big(b_J^s\big)_{J \subseteq N}$ specified above and the definition of the $(\underline{x}^s, \overline{x}^s)$-FBR, by eliminating redundant items, we have
\begin{align*}
f^s([P_1]^s, [P_2]^s, [P_3]^s)
=
{\footnotesize
\max\nolimits^{\prec^s}\!\!
\left(\!\!
\begin{array}{l}
\min\nolimits^{\prec^s}\!\!\big(\overline{x}^s, r_1([P_1]^s)\big),\;~~~~~~~~
\min\nolimits^{\prec^s}\!\!\big(r_1([P_1]^s), r_1([P_2]^s)\big),\\
\min\nolimits^{\prec^s}\!\!\big(r_1([P_1]^s), r_1([P_3]^s)\big), \; \min\nolimits^{\prec^s}\!\!\big(\underline{x}^s, r_1([P_2]^s), r_1([P_3]^s)\big)
\end{array}
\!\!\right). }
\end{align*}

If $r_1([P_1]^s) \in \langle \underline{x}^s, \overline{x}^s\rangle$
which implies $\min^{\prec^s}\!\big(\overline{x}^s, r_1([P_1]^s)\big) = r_1([P_1]^s)$,
we can further calculate that
$f^s([P_1]^s, [P_2]^s, [P_3]^s) = r_1([P_1]^s)$.

If $r_1([P_1]^s) \prec^s \underline{x}^s$ which implies $\min^{\prec^s}\!\big(\overline{x}^s, r_1([P_1]^s)\big) = r_1([P_1]^s)$, we can further calculate that
\begin{align*}
f^s([P_1]^s, [P_2]^s, [P_3]^s) = & \max\nolimits^{\prec^s}\!\!\Big(
r_1([P_1]^s), \min\nolimits^{\prec^s}\!\!\big(\underline{x}^s, r_1([P_2]^s), r_1([P_3]^s)\big)\!\Big)\\
= &~\textrm{med}^{\prec^s} \big(r_1([P_1]^s), r_1([P_2]^s), r_1([P_3]^s), r_1([P_1]^s), \underline{x}^s\big).
\end{align*}

If $\overline{x}^s \prec^s r_1([P_1]^s)$ which implies $\min^{\prec^s}\big(\overline{x}^s, r_1([P_1]^s)\big) = \overline{x}^s$, we can further calculate that
\begin{align*}
f^s([P_1]^s, [P_2]^s, [P_3]^s) = & \max\nolimits^{\prec^s}\!\!\Big(
\overline{x}^s, \min\nolimits^{\prec^s}\!\!\big(r_1([P_1]^s), r_1([P_2]^s)\big),
\min\nolimits^{\prec^s}\!\!\big(r_1([P_1]^s), r_1([P_3]^s)\big)\!\Big)\\
=& ~\textrm{med}^{\prec^s} \big(r_1([P_1]^s), r_1([P_2]^s), r_1([P_3]^s), r_1([P_1]^s), \overline{x}^s\big).
\end{align*}

This completes the verification of the claim.\medskip

Furthermore, for each $t \in M\backslash \{s\}$, we construct a dictatorial marginal SCF $f^t: \big[[\mathbb{D}]^t\big]^3 \rightarrow A^t$, where voter $3$ is fixed to be the dictator.
Then, we assemble the SCF $f: \mathbb{D}^3 \rightarrow A$ such that for all $(P_1, P_2, P_3) \in \mathbb{D}^3$,
$f(P_1, P_2,P_3) = \big(f^1([P_1]^1, [P_2]^1,[P_3]^1), \dots, f^m([P_1]^m, [P_2]^m, [P_3]^m)\big)$.
Clearly, by the decomposable-domain hypothesis, $f$ is a strategy-proof rule.
There are three subcases: (i) $x^s \in \langle \underline{x}^s, \overline{x}^s\rangle$,
(ii) $x^s \prec^s \underline{x}^s$ and (iii) $\overline{x}^s \prec^s x^s$.

In subcase (i), since $a^s \in \textrm{Int}\langle x^s, b^s\rangle$ and $a^s \notin \textrm{Int}\langle \underline{x}^s, \overline{x}^s\rangle$,
it is true that either $b^s \prec^s a^s \preccurlyeq^s \underline{x}^s \preccurlyeq^s x^s \preccurlyeq^s \overline{x}^s$ and $a^s \prec^s x^s$, or $\underline{x}^s \preccurlyeq^s x^s \preccurlyeq^s \overline{x}^s \preccurlyeq^s a^s \prec^s b^s$ and $x^s \prec^s a^s$.
Fix $P_2 = P_i$ and $P_1,P_2',P_3 \in \mathbb{D}$ such that $r_1(P_1) = b$, $r_1(P_2') = a$ and $r_1(P_3) = a$.
If $b^s \prec^s a^s \preccurlyeq^s \underline{x}^s \preccurlyeq^s x^s \preccurlyeq^s \overline{x}^s$ and $a^s \prec^s x^s$ hold, by Claim 1, we have
$f(P_1, P_2, P_3)^s = \textrm{med}^{\prec^s} \big(b^s, x^s, a^s; b^s, \underline{x}^s\big)
= a^s$ and
$f(P_1, P_2', P_3)^s = \textrm{med}^{\prec^s} \big(b^s, b^s, a^s; b^s, \underline{x}^s\big)
= b^s$.
Symmetrically, if $\underline{x}^s \preccurlyeq^s x^s \preccurlyeq^s \overline{x}^s \preccurlyeq^s a^s \prec^s b^s$ and $x^s \prec^s a^s$ hold, by Claim 1, we also have
$f(P_1, P_2, P_3)^s = \textrm{med}^{\prec^s} \big(b^s, x^s, a^s; b^s, \overline{x}^s\big)
= a^s$ and
$f(P_1, P_2', P_3)^s = \textrm{med}^{\prec^s} \big(b^s, b^s, a^s; b^s, \overline{x}^s\big)
= b^s$.
For all $t \in M\backslash \{s\}$, $f(P_1, P_2, P_3)^t = f(P_1, P_2', P_3)^t = r_1([P_3]^t) = z^t$.
Hence, $f(P_1, P_2, P_3) = (a^s, z^{-s}) = a$ and $f(P_1, P_2', P_3) = (b^s, z^{-s}) = b$
which by strategy-proofness imply $a\mathrel{P_2}b$, as required.


Subcases (ii) and (iii) are symmetric.
We focus on subcase (ii).
Since $a^s \in \textrm{Int}\langle x^s, b^s\rangle$ and $a^s \notin \textrm{Int}\langle \underline{x}^s, \overline{x}^s\rangle$,
it is true that either $b^s \prec^s a^s \prec^s  x^s \prec^s \underline{x}^s$,
or $x^s \prec^s \underline{x}^s \prec^s \overline{x}^s \preccurlyeq^s a^s \prec^s b^s$,
or $x^s \prec^s a^s \preccurlyeq^s \underline{x}^s$ and $x^s \prec^s a^s \prec^s b^s$.
If $b^s \prec^s a^s \prec^s  x^s \prec^s \underline{x}^s$
or $x^s \prec^s \underline{x}^s \prec^s \overline{x}^s \prec^s a^s \prec^s b^s$ holds,
similar to the verification in subcase (i),
given $P_2 = P_i$ and $P_1,P_2',P_3 \in \mathbb{D}$
such that $r_1(P_1) = b$, $r_1(P_2') = a$ and $r_1(P_3) = a$,
we have $f(P_1, P_2, P_3) = (a^s, z^{-s}) = a$ and $f(P_1, P_2', P_3) = (b^s, z^{-s}) = b$
which by strategy-proofness imply $a\mathrel{P_2}b$, as required.
If $x^s \prec^s a^s \prec^s \underline{x}^s$ and $x^s \prec^s a^s \prec^s b^s$ hold,
fix $P_1 = P_i$ and $P_1',P_2,P_3 \in \mathbb{D}$
such that $r_1(P_1') = b$, $r_1(P_2) = a$ and $r_1(P_3) = b$.
Then,we have
$f(P_1, P_2, P_3)^s = \textrm{med}^{\prec^s} \big(x^s, a^s, b^s, x^s, \underline{x}^s\big)=a^s$ and
$f(P_1', P_2, P_3)^s = \textrm{med}^{\prec^s} \big(b^s, a^s, b^s, b^s, \underline{x}^s\big)=b^s$ by Claim 1.
For all $t \in M\backslash \{s\}$,
$f(P_1, P_2, P_3)^t = f(P_1', P_2, P_3)^t = r_1([P_3]^t) = z^t$.
Hence, $f(P_1, P_2, P_3) = (a^s, z^{-s}) = a$ and $f(P_1', P_2, P_3) = (b^s, z^{-s}) = b$
which by strategy-proofness imply $a\mathrel{P_1}b$, as required.

In conclusion, $P_i$ is multidimensional hybrid on $\prec$ w.r.t.~$\underline{x}$ and $\overline{x}$.
\end{proof}

By Lemma \ref{lem:connectedgraph}, we know that
for each $s \in M$, $G_{\approx}^{A^s}$ is a connected graph,
which partly meets condition (ii) of Definition \ref{def:AMH}.
We complete the verification of condition (ii) of Definition \ref{def:AMH} in the last two lemmas.

\begin{lemma}\label{lem:implication}
Given $s \in M$, consider the induced marginal domain $[\mathbb{D}]^s$.
Given $a^s, b^s, c^s \in A^s$ such that $a^s \neq c^s$ and $b^s \neq c^s$,
if $b^s$ is included in all paths in $G_{\approx}^{A^s}$ connecting $a^s$ and $c^s$,
we have $b^s\mathrel{[P_i]^s}c^s$ for all $[P_i]^s \in [\mathbb{D}]^s$ with $r_1([P_i]^s) = a^s$.
\end{lemma}

\begin{proof}
Suppose by contradiction that it is not true, i.e., there exists
$[P_i]^s \in [\mathbb{D}]^s$ such that $r_1([P_i]^s) = a^s$ and $c^s\mathrel{[P_i]^s}b^s$.
Let $r_1([P_i]^t) = x^t$ for all $t \in M \backslash \{s\}$.
Thus, we know $r_1(P_i) = (a^s, x^{-s})$.
Since $c^s\mathrel{[P_i]^s}b^s$, we have $(c^s, x^{-s})\mathrel{P_i}(b^s, x^{-s})$.
We fix a preference $P_i' \in \mathbb{D}$ such that $r_1(P_i') = (c^s, x^{-s})$ by minimal richness.
Since $r_1(P_i) \neq r_1(P_i')$, by the Exterior\textsuperscript{+} property,
there exists a path $(P_{i|1}, \dots, P_{i|v})$ in $G_{\sim/\sim^+}^{\mathbb{D}}$ connecting $P_i$ and $P_i'$ that has no $\{(c^s, x^{-s}), (b^s, x^{-s})\}$-restoration.
Since $(c^s, x^{-s})\mathrel{P_i}(b^s, x^{-s})$ and $(c^s, x^{-s})\mathrel{P_i'}(b^s, x^{-s})$,
we know $(c^s, x^{-s})\mathrel{P_{i|k}}(b^s, x^{-s})$ for all $k = 1, \dots, v$ by no $\{(c^s, x^{-s}), (b^s, x^{-s})\}$-restoration.\medskip

\noindent
\textsc{Claim 1}: We have $r_1([P_{i|k}]^s) \neq b^s$ for all $k = 1, \dots, v$.\medskip

Suppose not, i.e., $r_1([P_{i|k}]^s) \neq b^s$ for some $k \in \{1, \dots, v\}$.
By the proof of Lemma \ref{lem:tops-onlydomain}, we know that all preferences of $\mathbb{D}$ are top-separable, which implies $(b^s, x^{-s})\mathrel{P_{i|k}} (c^s, x^{-s})$ - a contradiction.
\medskip

\noindent
\textsc{Claim 2}: For all $k \in \{1, \dots, v-1\}$, we have either $[P_{i|k}]^s = [P_{i|k+1}]^s$ or
$[P_{i|k}]^s \sim [P_{i|k+1}]^s$.\footnote{The notion of adjacency also works for the induced marginal preferences.}\medskip

Given $k \in \{1, \dots, v-1\}$, we consider the preferences $P_{i|k}$ and $P_{i|k+1}$.
Clearly, either $P_{i|k} \sim P_{i|k+1}$ or $P_{i|k} \sim^+ P_{i|k+1}$ holds.
First, let $P_{i|k} \sim P_{i|k+1}$.
Thus, there exist $y,z \in A$ such that
$r_{q}(P_{i|k}) = r_{q+1}(P_{i|k+1}) = y$ and $r_{q}(P_{i|k+1}) = r_{q+1}(P_{i|k}) = z$ for some
$1 \leq q < |A|$, and
$r_{\ell}(P_{i|k}) = r_{\ell}(P_{i|k+1})$ for all $\ell \notin \{q, q+1\}$.
If $y^{-s} = z^{-s} = x^{-s}$, then
$[P_{i|k}]^s$ and $[P_{i|k+1}]^s$ differ exactly on the relative rankings of $y^s$ and $z^s$, and hence
$[P_{i|k}]^s \sim [P_{i|k+1}]^s$, as required.
If $y^{-s} \neq  x^{-s}$ or $z^{-s} \neq x^{-s}$,
then $[P_{i|k}]^s = [P_{i|k+1}]^s$, as required.
Next, let $P_{i|k} \sim^+ P_{i|k+1}$.
Thus, there exist $\tau \in M$ and distinct $y^{\tau}, z^{\tau} \in A^{\tau}$ such that
\begin{itemize}
\item[\rm (1)] for each $z^{-\tau}\in A^{-\tau}$,
$(y^{\tau}, z^{-\tau}) = r_q(P_{i|k}) = r_{q+1}(P_{i|k+1})$ and
$(z^{\tau}, z^{-\tau})=r_q(P_{i|k+1}) = r_{q+1}(P_{i|k})$ for some $1\leq q < |A|$, and

\item[\rm (2)] \makebox{for all $d \in A$, $\big[d^{\tau} \notin \{y^{\tau}, z^{\tau}\}\big] \Rightarrow
\big[d = r_{\ell}(P_{i|k}) = r_{\ell}(P_{i|k+1})\;\textrm{for some} \; 1 \leq \ell \leq |A|\big]$.}
\end{itemize}
Clearly, either $\tau = s$ or $\tau \neq s$ holds.
If $\tau = s$,
$(y^s, x^{-s})$ and $(z^s, x^{-s})$ are consecutively and oppositely ranked across $P_{i|k}$ and $P_{i|k+1}$ by condition (1), while all alternatives of $(A^s, x^{-s})\backslash \{(y^s, x^{-s}),(z^s, x^{-s})\}$ are identically ranked in both $P_{i|k}$ and $P_{i|k+1}$ by condition (2).
Therefore, we have $[P_{i|k}]^s \sim [P_{i|k+1}]^s$, as required.
If $\tau \neq s$, then all alternatives of $(A^s, x^{-s})$ are identically ranked in both $P_{i|k}$ and $P_{i|k+1}$ by condition (2), and hence we have $[P_{i|k}]^s = [P_{i|k+1}]^s$, as required.
This completes the verification of the claim.
\medskip

Notice that by Claim 2, in the sequence $\big([P_{i|1}]^s, \dots, [P_{i|k}]^s, [P_{i|k+1}]^s, \dots, [P_{i|v}]^s\big)$,
some induced marginal preferences may appear multiple times.
For instance, let $[P_{i|q}]^s = [P_{i|q'}]^s$ where $1 \leq q < q' \leq v$.
By Claim 2, we know either $[P_{i|q'}]^s = [P_{i|q'+1}]^s$ or $[P_{i|q'}]^s \sim [P_{i|q'+1}]^s$.
Hence, either $[P_{i|q}]^s = [P_{i|q'+1}]^s$ or $[P_{i|q}]^s \sim [P_{i|q'+1}]^s$ holds.
Then, we remove $[P_{i|q+1}]^s, \dots, [P_{i|q'}]^s$, and refine the sequence to
$\big([P_{i|1}]^s, \dots, [P_{i|q}]^s, [P_{i|q'+1}]^s, \dots, [P_{i|v}]^s\big)$,
where every consecutive pair of the remaining induced marginal preferences remains to be either identical or adjacent.
By repeatedly eliminating repetitions of induced marginal preferences in the original sequence,
we finally elicit a subsequence of pairwise distinct induced marginal preferences
$\big([P_{i|k_1}]^s, \dots,[P_{i|k_q}]^s, [P_{i|k_{q+1}}]^s, \dots, [P_{i|k_w}]^s\big)$, where $w\geq 2$, such that
$[P_{i|k_1}]^s = [P_i]^s$, $[P_{i|k_w}]^s = [P_i']^s$ and $[P_{i|k_q}]^s \sim [P_{i|k_{q+1}}]^s$
for all $q \in \{1, \dots, w-1\}$.
Furthermore, let $r_1([P_{i|k_q}]^s) = x_q^s$ for each $q = 1, \dots, w$.
Clearly, $x_1^s = a^s$ and $x_w^s = c^s$.
Note that for all $q \in \{1, \dots, w-1\}$,
either $x_q^s = x_{q+1}^s$, or $x_q^s \neq x_{q+1}^s$.
Furthermore, if $x_q^s \neq x_{q+1}^s$, $[P_{i|k_q}]^s \sim [P_{i|k_{q+1}}]^s$ immediately implies
$x_q^s \approx x_{q+1}^s$.
Thus, in the sequence $(x_1^s, \dots, x_w^s)$,
for all $q \in \{1, \dots, w-1\}$, we have either $x_q^s = x_{q+1}^s$, or $x_q^s \approx x_{q+1}^s$.
Then, by a similar way of eliminating repeated elements,
we can induce a subsequence of pairwise distinct elements $(x_{q_1}^s, \dots, x_{q_o}^s, x_{q_{o+1}}^s, \dots, x_{q_r}^s)$, where $r \geq 2$,
such that $x_{q_1}^s = x_1^s = a^s$, $x_{q_r}^s = x_{w}^s = c^s$, and $x_{q_o}^s \approx x_{q_{o+1}}^s$ for all $o \in \{1, \dots, r-1\}$.
Thus, $(x_{q_1}^s, \dots, x_{q_r}^s)$ is a path in $G_{\approx}^{A^s}$ that connects $a^s$ and $c^s$.
Last, by Claim 1, we know $x_{q_o}^s \neq b^s$ for all $o = 1, \dots, r$.
Consequently, the path $(x_{q_1}^s, \dots, x_{q_r}^s)$ in $G_{\approx}^{A^s}$ connects $a^s$ and $c^s$, and excludes $b^s$.
This contradicts the hypothesis of statement (ii) that $b^s$ is included in all paths in $G_{\approx}^{A^s}$ connecting $a^s$ and $c^s$.
This proves the Lemma.
\end{proof}

\begin{lemma}\label{lem:noleaf}
Given $s \in M$, let $\underline{x}^s \neq \overline{x}^s$.
The subgraph $G_{\approx}^{\langle \underline{x}^s,\, \overline{x}^s\rangle}$ has no leaf.
\end{lemma}

\begin{proof}
For ease of presentation, let $A^s = \big\{a_1^s,\dots, a_{|A^s|}^s\big\}$, where
$a_k^s \prec^s a_{k+1}^s$ for all $k = 1, \dots, |A^s|-1$.
Clearly, $a_k^s\mathrel{[\underline{P}_i]^s}a_{k+1}^s$ and
$a_{k+1}^s\mathrel{[\overline{P}_i]^s}a_k^s$ for all $k = 1, \dots, |A^s|-1$.
Furthermore, by \textbf{Condition (b)}, we know $\underline{x}^s = a_{\underline{k}}^s$ and $\overline{x}^s = a_{\overline{k}}^s$ for some $1 \leq \underline{k}< \overline{k}\leq |A^s|$ such that $\overline{k}-\underline{k}>1$.
Then, by Lemma \ref{lem:connectedgraph} and \textbf{Condition (a)},
$G_{\approx}^{A^s}$ must be a combination of
the line $(a_1^s, \dots, a_{\underline{k}}^s)$,
the connected subgraph $G_{\approx}^{\langle a_{\underline{k}}^s,\, a_{\overline{k}}^s\rangle}$, and
the line $(a_{\overline{k}}^s, \dots, a_{|A^s|}^s)$.

Suppose by contradiction that $G_{\approx}^{\langle a_{\underline{k}}^s,\, a_{\overline{k}}^s\rangle}$ has a leaf.
Thus, we have $a_k^s, a_{k'}^s \in \langle a_{\underline{k}}^s,\, a_{\overline{k}}^s\rangle$
such that $a_k^s$ is a leaf of $G_{\approx}^{\langle a_{\underline{k}}^s,\, a_{\overline{k}}^s\rangle}$ and
$a_k^s \approx a_{k'}^s$.
There are two cases: $\underline{k}<k< \overline{k}$ and
$k \in \{\underline{k}, \overline{k}\}$.
In each case, we induce a contradiction.
First, let $\underline{k}<k< \overline{k}$.
Thus, $a_{k'}^s$ is included in all paths in $G_{\approx}^{A^s}$ connecting $a_1^s$ and $a_k^s$, and that $a_{k'}^s$ is also included in all paths in $G_{\approx}^{A^s}$ connecting $a_{|A^s|}^s$ and $a_k^s$.
Immediately, Lemma \ref{lem:implication} implies $a_{k'}^s\mathrel{[\underline{P}_i]^s}a_k^s$ and
$a_{k'}^s\mathrel{[\overline{P}_i]^s}a_k^s$, which contradict the fact that $[\underline{P}_i]^s$ and $[\overline{P}_i]^s$ are complete reversals.
Henceforth, let $k \in \{\underline{k}, \overline{k}\}$.
We assume w.l.o.g.~that $k = \underline{k}$.
The verification related to $k = \overline{k}$ is symmetric.\medskip

\noindent
\textsc{Claim 1}:
We have $\underline{k}<k'< \overline{k}$.
\medskip

It is clear that $\underline{k}<k' \leq \overline{k}$.
Suppose by contradiction that $k' = \overline{k}$.
Since $\overline{k}-\underline{k}>1$,
we fix $a_p^s \in \langle a_{\underline{k}}^s,\, a_{\overline{k}}^s\rangle$ such that $\underline{k}<p< \overline{k}$.
Since $G_{\approx}^{\langle a_{\underline{k}}^s,\, a_{\overline{k}}^s\rangle}$ is a connected subgraph,
$a_{\underline{k}}^s$ is a leaf and $a_{\underline{k}}^s \approx a_{\overline{k}}^s$,
it must be the case that $a_{\overline{k}}^s$ is included in all paths in $G_{\approx}^{\langle a_{\underline{k}}^s,\, a_{\overline{k}}^s\rangle}$ connecting $a_{\underline{k}}^s$ and $a_p^s$.
Consequently, $a_{\overline{k}}^s$ is included in all paths in $G_{\approx}^{A^s}$ connecting $a_1^s$ and $a_p^s$.
Immediately, Lemma \ref{lem:implication} implies $a_{\overline{k}}^s\mathrel{[P_i]^s}a_p^s$ for all
$[P_i]^s \in [\mathbb{D}]^s$ with $r_1([P_i]^s) = a_1^s$.
However, we have $r_1([\underline{P}_i]^s) = a_1^s$ and $a_p^s\mathrel{[\underline{P}_i]^s}a_{\overline{k}}^s$ - a contradiction.
This completes the verification of the claim.
\medskip

Now, since $G_{\approx}^{\langle a_{\underline{k}}^s,\, a_{\overline{k}}^s\rangle}$ is a connected subgraph,
$a_{\underline{k}}$ is a leaf and $a_{\underline{k}}^s \approx a_{k'}^s$,
it is evident that $G_{\approx}^{\langle a_{\underline{k}}^s,\, a_{\overline{k}}^s\rangle\backslash \{a_{\underline{k}}^s\}}$ is also connected subgraph.
Consequently, $G_{\approx}^{A^s}$ is a combination of
the line $(a_1^s, \dots, a_{\underline{k}}^s,a_{k'}^s)$, the connected subgraph $G_{\approx}^{\langle a_{\underline{k}}^s,\, a_{\overline{k}}^s\rangle\backslash \{a_{\underline{k}}^s\}}$ and
the line $(a_{\overline{k}}^s, \dots, a_{|A^s|}^s)$.
There are two subcases: $\overline{k}-\underline{k} = 2$ and
$\overline{k}-\underline{k} > 2$.
If $\overline{k}-\underline{k} = 2$,
it is evident that $k' = \underline{k}+1$.
Then, $G_{\approx}^{A^s}$ degenerates to the line $\big(a_1^s, \dots, a_{\underline{k}}^s, a_{\underline{k}+1}^s, a_{\underline{k}}^s, \dots, a_{|A^s|}^s\big)$.
Consequently, by Lemma \ref{lem:implication},
all induced marginal preferences of $[\mathbb{D}]^s$ are single-peaked on $\prec^s$, and hence
equivalently hybrid on $\prec^s$ w.r.t.~$a_{\underline{k}}^s$ and $a_{\overline{k}+1}^s$.
This contradicts \textbf{Condition (b)}.
Last, let $\overline{k}-\underline{k} > 2$.
We first construct a linear order $\lhd^s$ over $A^s$ such that
(1) $\big[1 \leq p < p' \leq \underline{k}\; \textrm{or}\; \overline{k} \leq p < p' \leq |A^s|\big] \Rightarrow [a_p^s \lhd^s a_{p'}^s]$,
(2) $a_{\underline{k}}^s \lhd^s a_{k'}^s$, and
(3) $\min^{\lhd^s} \{a_{\underline{k}+1}^s, \dots, a_{\overline{k}}^s\}= a_{k'}^s$ and
$\max^{\lhd^s} \{a_{\underline{k}+1}^s, \dots, a_{\overline{k}}^s\}= a_{\overline{k}}^s$.
We then fix distinct marginal thresholds $a_{k'}^s$ and $a_{\overline{k}}^s$ on $\lhd^s$.
Clearly, $\langle a_{k'}^s, a_{\overline{k}}^s\rangle^{\lhd^s} = \langle a_{\underline{k}}^s, a_{\overline{k}}^s\rangle\backslash \{a_{\underline{k}}^s\} \subset \langle a_{\underline{k}}^s, a_{\overline{k}}^s\rangle$.
Since $G_{\approx}^{A^s}$ is a combination of
the line $(a_1^s, \dots, a_{\underline{k}}^s,a_{k'}^s) = \langle a_1^s, a_{k'}^s\rangle^{\lhd^s}$, the connected subgraph $G_{\approx}^{\langle a_{\underline{k}}^s,\, a_{\overline{k}}^s\rangle\backslash \{a_{\underline{k}}^s\}} = G_{\approx}^{\langle a_{k'}^s,\, a_{\overline{k}}^s\rangle^{\lhd^s}}$ and
the line $(a_{\overline{k}}^s, \dots, a_{|A^s|}^s) = \langle a_{\overline{k}}^s, a_{|A^s|}^s\rangle^{\lhd^s}$, Lemma \ref{lem:implication} implies that
all induced marginal preferences of $[\mathbb{D}]^s$ are hybrid on $\lhd^s$ w.r.t.~$a_{k'}^s$ and $a_{\underline{k}}^s$.
This contradicts \textbf{Condition (b)} as well.
This proves the Lemma.
\end{proof}

This completes the verification in \textbf{Part 3}, and hence proves the Theorem.

\end{document}